\documentclass[11pt]{article} 

\usepackage{amsmath,amsthm,latexsym,amssymb,amsfonts,epsfig}

\newtheorem{Theorem}{Theorem}[section]

\newcommand{\be}{\begin{equation}}
\newcommand{\ee}{\end{equation}}
\newcommand{\ba}{\begin{eqnarray}}
\newcommand{\ea}{\end{eqnarray}}

\title{{\sf Non-degenerate metrics, hypersurface deformation algebra,}\\
 {\sf non-anomalous representations and density weights in quantum gravity}} 
\author{
{\sf T. Thiemann}$^1$\thanks{{\sf 
thomas.thiemann@gravity.fau.de}}\\
\\
{\sf $^1$ Inst. for Quantum Gravity, FAU Erlangen -- N\"urnberg,}\\
{\sf Staudtstr. 7, 91058 Erlangen, Germany}\\
}
\date{{\small\sf \today}}

\makeatletter
\@addtoreset{equation}{section}
\makeatother

\begin{document} 

\maketitle

{\sf

\begin{abstract}
Classical General Relativity is a dynamical theory of spacetime 
metrics of Lorentzian signature. In particular the classical metric 
field is nowhere degenerate in spacetime. In its initial value formulation with 
respect to a Cauchy surface the induced metric is of Euclidian signature 
and nowhere degenerate on it. It is only under this assumption of 
non-degenaracy of the induced metric that one can derive the hypersurace 
deformation algebra between the initial value constraints which is 
absolutely transparent from the fact that the {\it inverse} 
of the induced metric is needed to close the algebra. This statement 
is independent of the density weight that one may want to equip 
the spatial metric with.

Accordingly, the very definition of a non-anomalous representation of 
the hypersurface defomation algebra in quantum gravity 
has to address the issue of 
non-degenracy of the induced metric that is needed in the classical theory.
In the Hilbert space representation employed in Loop Quantum Gravity
(LQG) most emphasis has been layed to define an inverse metric operator 
on the dense domain of spin network states although they represent induced 
quantum geometries which are degenerate almost everywhere. It is no surprise 
that demonstration of closure of the constraint algebra on this domain meets 
difficulties because it is a sector of the quantum theory which is 
classically forbidden and which lies outside the domain of definition of the 
classical hypersurface deformation algebra. Various suggestions for addressing
the issue such as non-standard operator topologies, dual spaces (habitats) 
and density weights have been propposed to address this issue
with respect to the quantum dynamics of LQG.

In this article we summarise these developments and argue that insisting on 
a dense domain of non-degenerate states within the LQG representation 
may provide a natural 
resolution of the issue thereby possibly avoiding the above mentioned 
non-standard constructions.      
\end{abstract}

\section{Introduction}
\label{s1}

One possible approach to quantum gravity is via the Hamiltonian or canonical
formulation \cite{1}. This so called initial value 
formulation is widely used in mathematical general
relativity \cite{2} and numerical relativity \cite{3} with recent 
spectacular success in predicting e.g. black hole merger templates
\cite{4}. The canonical approach is also the fundament of Loop Quantum Gravity 
(LQG) \cite{5}. LQG derives its name from the fact that GR can be formulated 
in terms of Yang-Mills like non-Abelian connection variables \cite{6} 
and thus methods from lattice QCD \cite{7} (specifically Wilson {\it loops})
are employed in the quantisation.

A central ingredient of the Hamiltonian approach is the abstract hypersurface 
deformation algebra $\mathfrak{h}$. As shown 
in the seminal paper \cite{8}, every generally covariant Lagrangian (e.g. GR
with any type of matter) has a singular Legendre transform, leading 
to initial value constraints whose Poisson algebra is a representation of 
$\mathfrak{h}$. The algebra $\mathfrak{h}$ is isomorphic to the Lie algebra 
of spacetime diffeomorphisms when the equations of motion (e.g. Einstein
equations and Bianchi identities) hold. As an abstract algebra it can 
be defined as follows: Under the usual 
assumption of global hyperbolicity without which the initial value 
formualtion is ill-posed, the spacetime manifold $M$ is diffeomorphic
to $\mathbb{R}\times \sigma$. The freedom in choosing this diffeomorphism 
and thus setting up an initial value formulation is encoded by a scalar 
(``lapse'') function $n\in N$ and a (``shift'') vector field $u\in U$ 
on $\sigma$
respectively which depend parametrically on the time coordinate $t$ that 
defines a foliation of $M$ into leaves $\Sigma_t$ which are all diffeomorphic 
to $\sigma$. Then the fundamental $\mathfrak{h}$ algebra reads (we absorb 
any constant such as Newton's constant into $C,D$)
\be \label{1.1}
\{D(u),D(v)\}=-D([u,v]),\;\;
\{D(u),C(n)\}=-C(u[N]),\;\;
\{C(m),C(n)\}=-D(q^{-1}\;[M\; dN-N\; dM])
\ee
Here $[u,v]$ is the lie bracket of vector fields and $u[N]$ the action of 
$u$ considered as a derivation on the scalar functions. The first two 
relations in (\ref{1.1}0 
therefore just depend on the differentiable structure on $\sigma$. 
The last relation of (\ref{1.1}) however also depends on a metric tensor 
field $q$ on $\sigma$ whose {\it inverse} features into (\ref{1.1}). 
It is therefore a prerequisite for the {\it very definition} of 
$\mathfrak{h}$ that $q$ be {\it invertible everywhere} on $\sigma$ and 
{\it at all times}. In GR the physical meaning of $q$ is of course to be 
the pullback of the spacetime metric $g$ to $\sigma$ (at the respective time).
Because of this, (\ref{1.1}) is strictly speaking not a Lie algebra because 
its right hand side cannot be expressed as linear combinations of $C,D$ 
with (so called structure) constant coefficients, rather those coefficients 
are themselves functions of the dynamical fields. It is therefore 
customary to call them structure functions rather than structure constants.   
 
A classical representation  of (\ref{1.1}) as a Poisson bracket algebra is 
generated 
by a generally covariant Lagrangian such as GR with any matter coupling
including a possible cosmological constant. In this case $C,D$ acquire 
the meaning of Hamiltonian and spatial diffeomorphism constraint respectively.
Their Poisson brackets with the fields underlying the canonical formulation
are equivalent to the Lagrangian equations of motion (e.g. Einstein 
equations and Bianchi identities). In quantum gravity, one is interested in 
a quantum representation of (\ref{1.1}) by commutators of operators
defined on a common (i.e. independent of $N,u$), dense and invariant domain
$\cal D$ of a Hilbert space $\cal H$ which is supposed to implement the 
canonical (anti-) commutation relations among the the fields (such as 
$q$ and its conjugate moemntum $p$) such that
(we absorb $\hbar$ into $C,D$)
\be \label{1.2}
[D(u),D(v)]=-i\;D([u,v]),\;\;
[D(u),C(n)]=-i\;C(u[N]),\;\;
[C(m),C(n)]=-i\;D(q^{-1}\;[M\; dN-N\; dM])
\ee
where $D(x),\; C(x);\; x\in \sigma$ 
together with $q(x)$ and $q(x)^{-1}$ have become operator valued 
distributions, $C(n)=\int_\sigma\; dx n(x)\; C(x)$ and similar for
$D(u)$. Therefore 
(\ref{1.2}) is not completely defined in terms of $C,D$ alone but also 
requires in addition information about what should be done with the 
operator valued distribution $q^{-1}$, e.g. ordering issues have to be 
addressed. In particular, the appearance of the {\it inverse} of $q$
requires that both $q,q^{-1}$ are well defined on ${\cal D}$. Moreover,
the {\it quantum Einstein} or Wheeler-DeWitt equations \cite{9} are 
the conditions on
distributions (linear functionals) $l$ on $\cal D$ such that 
\be \label{1.3}
l[C(n)\psi]=l[D(u)\psi]=0
\ee
for all $n\in N,\; u\in U\;\psi\in D$. Then the validity of (\ref{1.2}),
without anomalous terms not in the linear span of $C,D$,  
are ``integrability conditions'' in order that (\ref{1.3}) holds.  
The solutions of (\ref{1.3}) are in general not elements of $\cal H$ unless 
zero is in the joint pure point spectrum of all $C(n),\; D(n)$. To find 
an inner product on these ``generalised zero eigenvectors'' is then an 
additional task that one has to carry out. If (\ref{1.2}) was a true 
Lie algebra one could use the theory of rigged Hilbert spaces 
\cite{10}. Since it is not, one has to resort to different methods, for 
example \cite{11} which replaces the set of all $C(n),D(u)$ by a single 
(``master'') constraint $M$ to which the theory of rigged Hilbert spaces 
may then be applied.\\
\\
This rough sketch of the general canonical quantisation programme now must
be implemented concretely. This is what has been done in LQG. Inspired by 
the fact that fermionic matter requires to work with Vielbeins and connections 
rather than metrics, a connection formulation \cite{6} has been introduced.
A rigorous Hilbert space representation of the CCR for geometry \cite{11} and 
the matter CCR and CAR
\cite{12} has been defined which is designed to formulate the quantum 
dynamics non-perturbatively given the perturbative non-renormalisablility of 
quantum gravity and thus is not one of the standard Fock representations.
The difference between the LQG representation and the standard Fock 
representation can be seen for instance in the fact that the former 
is non not separable, that only finite but no infinitesimal unitary quantum 
diffeomorphisms $U(\varphi)$ can be defined \cite{12a}  (technically, 
1 parameter groups $s\mapsto \varphi^u_s$ of diffeomorphisms 
are not strongly continuous), that
area and volume operators \cite{13} can be defined at 
all and that their spectra are pure point and that connections themselves 
are ill defined while their holonomies are bounded operators. 
That these geometrical operators 
and their generalisations such as triads \cite{14} can be defined in the 
LQG representation turns out to be very important in order to define 
the quantum dynamics, i.e. the operators $C(n)$ \cite{15}, because 
literally all geometry and matter matter contributions including the 
cosmological constant depend on them. 

The definition \cite{15} of the operators $c(n)$ on the dense domain $\cal D$ 
given by the span of an ONB of $\cal H$ known as spin network functions
(SNWF) \cite{16} involves two steps: First a regularisation introducing 
a UV regulator $\epsilon$ necessary in order that the connections involved
in $c(n)$ are replaced by holonomises which yields operators 
$c_\epsilon(n)$. Secondly, the limit $\epsilon\to 0$ 
is taken in an operator topology which makes use of the existence of 
a space diffeomorphism invariant distributions $L$ on $\cal D$
proved in \cite{12a} resulting 
in regulator free operators $c(n)$. More in detail
\be \label{1.3a}
c_\epsilon(n)\to c(n)\;\;
\Leftrightarrow\;\;
l[c_\epsilon(n)\psi]\to l[c(n)\psi]\;\;\forall\; l\in L,\;\psi\in {\cal D}
\ee
It should be stressed that these 
operators are still defined on $\cal D$, that their commutator does not 
vanish but that their commutator annihilates $L$ and that 
its commutator with finite diffeomorphisms $U(\varphi)$ 
yields another operator $c(n_\varphi)$ up to a diffeomorphism. Stricly 
speaking therefore (\ref{1.2}) is {\it not} imeplemented, it cannot, as 
the operator $D(u)$ does not exist. However, (\ref{1.2}) is replaced by
\ba \label{1.4}
U(\varphi)\;U(\varphi') &=&
U(\varphi\circ \varphi')
\nonumber\\
U(\varphi)\;C(n)\;U(\varphi)^{-1} &=&
=U(\phi_{\varphi,n})\;C(\varphi^\ast n)
\nonumber\\
{[}C(m),C(n)] &=& -i\;
\sum_\varphi\; [U(\varphi)-1_{{\cal H}}],\;F'(m,n,\varphi)\; 
\ea
Here $\phi_{\varphi,n}$ is a diffeomorphism that depends on both $\varphi,n$
and $F'(m,n,\varphi)=-F'(n,m,\varphi)$ are ``structure operators''. The 
sum is formally over all diffeomorphisms but the structure operators vanish 
except for finitely many when acting on $\cal D$.
The algebra (\ref{1.4}) is therefore 
consistent in the sense that the joint 
kernel of the set of commutators of its generators 
is contained in the joint kernel of the set of its generators which is 
a minimal requirement, we call it  {\it mathematical anomaly freeness} in what
follows.  

At first this appears to be as close as one may hope to get given that 
$D(u)$ is not at one's disposal and that (\ref{1.4}) is a suitable substitute
for (\ref{1.2}) under these circumstances. In fact, the optimum that one 
could hope for would be  
\be \label{1.5}
U(\varphi)\;U(\varphi')=
U(\varphi\circ \varphi'),\;\;
U(\varphi)\;C(n)\;U(\varphi)^{-1}
=C(\varphi^\ast n),\;\;
[C(m),C(n)]=-i\;
\sum_\varphi\; F(m,n,\varphi)\; [U(\varphi)-1_{{\cal H}}],\;\
\ee
where the ``structure operators'' $F(m,n,\varphi)$ 
qualify as quantisations of the structure functions $q^{-1}(m\; dn-n\; dm)$.
The actual situation (\ref{1.4}) agrees with the wish list (\ref{1.5}) as far 
as the first relation is concerned, they agree with respect to the second 
up to a diffeomorphism but the last relation in (\ref{1.4}) is manifestly 
violated: The $F'$ do not qualify as a possible realisation of the $F$. 
To see this one quantises the operator corresponding to $D(q^{-1}(m dn-n dm)$
independently \cite{17} and notices that the structure operators $F,F'$ 
differ. To bring them to match, the second action by $c(m)$ would need
to be non-trivial at the new excitations of a SNWF that a first action by 
$c(n)$ has created. However, due to the properties of the volume operator 
used in the quantisation of $c(n)$ in \cite{15} that second action is trivial.
Changing that property in the current set-up would create even a 
{\it mathematical anomaly}.\\ 
\\
One can summarise the situation therefore by saying that \cite{15} is free 
of mathematical anomalies but that it displays a {\it physical anomaly} 
in the sense just described, namely that the algebra of generators 
closes under taking commutators, however with the wrong structure functions.
This is of course inacceptable: Using a crude ananolg from the theory 
of finite dimensional 
Lie algebras, suppose that one were to quantise the Lie algebra $A$ 
but ended up with a representation of the Lie algebra $B$ of the same 
dimension. Then one 
would wrongly 
define physical states to be those that are invariant under the 
symmetries (gauge transformations) generated by $B$ rather than those of 
$A$ (consider the examples of $A=$so(1,3) and $B=$so(4)).

To improve the situation one must obviously modify $c(n)$ within the 
LQG representation 
or, more drastically, change the LQG representation and 
start from the beginning. In this paper we focus on efforts 
that do not change the LQG representation but mention that currently
Hamiltonian renormalisation methods are being developed \cite{18}         
which potentially will change the LQG representation. These
non representation changing techniques
are based on so-called ``habitats'' \cite{19}. A habitat is a space 
of distributions $L'$ on $\cal D$ which is different from the space 
$L$ of diffeomorphism invariant distributions. One tries to define a dual 
representation $C'(n)$ of $C(n)$ on $L'$ by taking 
\be \label{1.6}
[C'(n) l'](\psi):=\lim_{\epsilon \to 0}\; l'[C_\epsilon(n)\psi]     
\ee
for all $\psi\in {\cal D}$. Provided the limit exists (in the sense
of complex numbers) it defines a new distribution on $\cal D$ and if the 
space $L'$ is carefully chosen, then $L'$ is invariant under the $C'(n)$. 
In that case one can take commutators and compute the algebra of dual 
operators. 

Note the difference between (\ref{1.3a}) and (\ref{1.6}):
In (\ref{1.3a}) the space $L$ is used to define a topology on the 
unbounded operators on $\cal H$ with common, dense, invariant domain $\cal D$
of the Hilbert space $\cal H$ while in (\ref{1.6}) one defines    
a new space of operators on a different space $L'$ which however is no 
Hilbert space, it does not come with an inner product. Therefore, the 
algebra of the $c'(n)$ is purely algebraic and is not equipped with any 
obvious topology while the $c(n)$ come equipped with any of the topologies 
that descend from $\cal H$ and $\cal D$. Note also that by definition 
$c(n)$ does not commute with diffeomorphisms and thus $L$ itself 
does not qualify as an invariant space $L'$ while it maybe a subspace thereof.

In \cite{19} such an invariant space $L'$ was found and it was shown that 
the algebra of the $C'(n)$ is Abelian on this $L'$.
Unfortunately, the difference between the $C(n)$ and the $C'(n)$ has 
been confused repeatedly in the literature \cite{20} and has wrongly 
lead to the statement that the $C(n)$ only close ``on shell'' (i.e. on
the space $L$ annihilated by the operators 
$U(\varphi)-1_{{\cal H}}$) and commute. Not only is this mathematically 
impossible 
because the $C(n)$ or rather their duals $C'(n)$
do not preserve $L$ but it is also technically wrong as the 
$C(n)$ are defined densely on $\cal D$, hence they are defined 
manifestly ``off shell'' and there their 
commutator does not vanish. See the extensive discussion in the second book 
of \cite{5} and the recent review \cite{18}.

To work with non-trivial habitats appears to be attractive because in the 
topology defined by (\ref{1.3a}) the resulting limit operators $C(n)$ carry 
a large amount of quantisation ambiguities. While the joint kernel 
of the $C(n),\; U(\varphi)-1_{{\cal H}}$ that one is eventually interested in 
is insensitive to most of these 
ambiguities it maybe hoped for that the $C'(n)$ have less ambiguities. 
Therefore an ambiguous programme has been recently launched \cite{21,21a}   
which also aims to find a non-trivial representation of the algebra 
generated by the $C'(n)$ such that $[C'(m),C'(n)]$ has no physical anomaly.
A common feature of these developments is that $C'(n)$ 
is replaced by 
$\tilde{C}'(n)$t
corresponding to quantisations (on certain dual spaces $L'$) of classical 
$\tilde{C}(n)$ wherea $\tilde{C}$ has {\it non-standard density weight}. 
This non-standard density weight is absolutely crucial in order to avoid 
the Abelian charactler of the algebra of the $C'(n)$. 

It was shown in 
\cite{21} that non-trivial choices of $L'$ exist which make the 
algebra of the $\tilde{C}'(n)$ for Euclidian signature vacuum GR without 
cosmological constant free of physical anomalies. This is a most astonishing 
and non-trivial result. It triggers the following questions:
\begin{itemize}
\item[1.] The detailed action of $\tilde{C}(n)$ at finite regulator 
and the choice of $L'$ have to be matched carefully to each other in 
order that the regulator can be removed for the dual operator
$\tilde{C}'(N)$ on $L'$ and its action be non-anomalous.
This raises the question how many suitable 
such choices do exist and what the residual amount of ambiguity is.
\item[2.] While for density weight unity the norm of the states in the 
image of the operator $C(N)$ converge as the regulator is removed, 
for the non standard density operator $\tilde{C}(N)$ that limit 
is divergent. This is the price to pay in order that the 
commutator of the $\tilde{C}'(N)$ be non-Abelian on 
the chosen $L'$. While one may argue 
that the $\tilde{C}'(N)$ are ``more fundamental'' than the 
$\tilde{C}(N)$ one may wonder whether one cannot have natural 
density weight unity without Abelian $C'(N)$ on suitable different $L'$. 
\item[3.] The non-standard density weight does not 
allow to switch to physical Lorentzian signature, or admit a 
non-vanishing cosmological constant or non-trivial matter coupling
as also stated in \cite{15,21}. We will review the reasons for 
this in section \ref{s2}. Thus one presumably needs to combine 
\cite{15,21} in some non-trivial way and resort to density weight one. 
\end{itemize}
Note that the whole strategy of solving the constraints in quantum 
theory rather than classically (gauge fixing) maybe criticised as being
beyond practicability for complicated theories such as GR: After all,
the quantum constraints have to be regularised and densely defined, the 
regulator removed, solutions to be found, the solutions to be equipped 
with a new inner product, observables to be defined on that new physical 
Hilbert space. All of that can be avoided using gauge fixing the classical 
theory and one works directly just with observables and the physical Hilbert 
space. That gauge fixings are usually plagued by global issues (Gribov 
copies) appears to be higher order problem given the immense dificulties 
in solving the constraints in quantum theory. Yet, the concerns about 
anomaly freeness of the constraints can not entirely be ignored in such a 
reduced phase space approach. This is because pieces of 
the constraints are building blocks of the physical Hamiltonian that drives 
the physical time evolution of observables and in that sense their quantisation 
ambiguities reappear, it is just that the anomaly issue is absent. 
On the other hand, typically gauge 
fixing conditions are phrased as coordinate conditions on scalar configuration
degrees of freedom $q$ and one then solves the constraints for the respective 
conjugate momenta $p=-h$ which are scalar densities of weight one. 
Thus we see that even in the gauge fixed theory density weight unity is 
the natural choice. See \cite{21b} for natural implementations of 
gauge fixing and the naturality of the density weight one choice.\\
\\
In this article we wish to complement the debate about an anomaly free 
implementation of $\mathfrak{h}$ by communicating the following observations:
\begin{itemize}
\item[1.] That, with minimal physical assumptions, 
the standard density weight of 
\cite{15} is the only viable one.
\item[2.] That the algebra 
of the $C'(n)$ of \cite{19} is Abelian on the chosen $L'$ 
is {\it physically correct}.
\item[3.] That the reason for this Abelian character is the neglectance 
of the implicit assumption about $\mathfrak{h}$, namely that $q$ be 
invertible.  
%s
\item[4.] That a non-Abelian algebra of the $C'(n)$ on suitable $L'$ 
and maybe even of the $C(n)$ on suitable $\cal D$ maybe possible with 
standard density weight on states on which $q$ is invertible.
\end{itemize}
The presumptions expressed in item [4.] rest on section \ref{s6} of the 
present paper and on our companion paper \cite{50} where it is shown that
when quantum non-degeneracy is taken into account the apparent 
tension between density unity and non-trivial quantum $\mathfrak{h}$ 
disappears at least in those theories. \\  
\\
The architecture of this article is as follows:\\
\\
In section \ref{s2} we show, independent of the arguments of 
\cite{11}, why an LQG like representation in quantum gravity is {\it 
dynamically preferred.} We also repeat from \cite{15}
why density weight one of $C$ is the only viable choice under 
physically well motivated assumptions and using again 
{\it dynamical input}.

In section \ref{s3} we show why current calculations of the algebra
of the $C(n)$ or $C'(n)$ are inconclusive and the result of the 
computations of \cite{19} are not at all surprising: Current calculations 
are performed in sectors of $\cal H$ or spaces of distributions $L'$ 
which represent quantum geometries which do not qualify as quantisations 
of classical geometries in which the classical algebra $\mathfrak{h}$
is defined. This is because current calculations are performed in regimes 
with quantum geometries that are degenerate almost everywhere.
In other words, we promote the point of view {\it that 
quantum geometrical non-degeneracy be a central ingredient in the 
very definition of $\mathfrak{h}$.} In doing so, the Abelian character 
of the dual action on certain spaces of distributions \cite{19} may
disappear by itself, without changing the density weight, thereby
avoiding the issues mentioned above. Also, many of the results of \cite{15}
can probably be transferred to the non-degenarate sector.

In section \ref{s5} we show how the non-degeneracy condition can potentially be taken 
care of in the Hamiltonian renormalisation of LQG \cite{18}. 
This programme is still its infancy. In a preliminary calculation we 
consider a certain set of
coherent non-degenrate states based on \cite{22} and compute 
expectation values of the commutators $[C(m),C(n)]$ with $C(m)$ chosen 
as in \cite{32}. Exploiting the freedom in the choice of these states 
one can get the expectation value of the difference between 
the commutator and the quantisation of the Poisson bracket 
as small as one wishes. 

In section \ref{s6} we consider the toy model of parametrised field theory
\cite{24}. It has been quantised by LQG methods before and it was 
demonstrated to have a non-anomalous algebra with non-standard 
density weight on some $L'$
\cite{25} as above and standard density weight 
on $\cal H$ \cite{26} if an addition a renormalisation step is invoked.
Here we employ a new LQG like representation which is better geared towards 
quantum non-degeneracy and we show that $C'(n),D'(u)$ with {\it standard 
density weight} can be represented on a certain space $L'$ without anomalies 
(but including the central charge of the Virasoro algebra). 

In section \ref{s7} we summarise and conclude.\\
\\
In closing, we stress that this is mostly conceptual work. We delibaratively
neglect many technical details in order not to draw attention away
from the main line of argument. However, all the technical details can 
be found in the original articles quoted along with.

\section{CCR and CAR representations of LQG type}
\label{s2}

In any QFT is quite important that one finds a ground state of the 
corresponding Hamiltonian in order that the dynamics defined by 
it can be constructed. For instance, choosing a Fock representation 
not precisely geared to the Hamiltonian $H$ of a free Klein Gordon field of 
mass $M$ makes the quantum dynamics ill defined. If one considers the 
Hamiltonian constraints $C(n)$ of GR (with or without matter) 
and looks for a representation of the CCR and CAR on a Hilbert space 
such that $C(n)$ 
be densely defined on a suitable common invariant domain $\cal D$ thereof
it is well motivated to try to find a representation based on a cyclic 
vector $\Omega$ whose excitations create $\cal D$. A pecularity of $C(n)$ 
is that every single piece of it depends non-trivially on the induced 
metric $q$ and/or its inverse (bosons) or the co-triad $e$
(square root of $q$) or its inverse. Accordingly, one has good chances 
to have a ground state of the $C(n)$ at one's disposal if one manages 
to build a representation based on a cyclic state $\Omega$ annhilated by 
$e$ or suitable aggregates formed from it. In LQG one considers the variables 
\be \label{2.1}
E^a_j:=\sqrt{\det(q)}\; e^a_j,\; q_{ab}=\delta_{jk} \;e_a^j\; e_b^k,\;\;
e^a_j\; e_a^k=\delta_j^k
\ee
where $a,b,c,..=1,2,3$ are tensor indices w.r.t. $\sigma$ while 
$a,b,c,..=j,k,l$ are tensor indices w.r.t. su(2). Canonically conjugate to
$E^a_j$ is a su(2) connection $A_a^j$ \cite{6}
which captures information about 
the extrinsic curvature of the Cauchy surfaces
\be \label{2.2}
\{E^a_j(x),\;A_b^k(y)\}=\delta^a_b\;\delta_j^k\;\delta(x,y)
\ee
where we have set the Newton constant equal to unity. 

To find a representation of the CCR corresponding to (\ref{2.2}) we 
consider the Weyl algebra defined by the Weyl elements 
\be \label{2.3}
W(f,F):=\exp(i\;[<E,f>+<F,A>]),\;\;
<E,f>:=\int_\sigma\; d^3x\;E^a_j(x)\; f_a^j(x),\;\;
<F,A>:=\int_\sigma\; d^3x\;F^a_j(x)\; A_a^j(x),\;\;
\ee
where we leave the nature of the test functions $f,F$ unspecified for the 
moment. Then we define $\cal D$ as the linear span of the $W(f,F)\Omega$
where $\Omega$ is annihilated by $E$ as motivated above, that is
\be \label{2.4}
W(f,0)\Omega=\Omega
\ee
It then follows immediately from the Weyl relations 
\be \label{2.5}
<\Omega,\; W(0,F)\;\Omega>_{{\cal H}}=
<\Omega,\;W(f,0)\; W(0,F)\;W(f,0)^{-1}\;\Omega>_{{\cal H}}
=\exp(-i<F,f>)\; <\Omega,\; W(0,F)\;\Omega>_{{\cal H}}
\ee
for any $F,f$ therefore automatically
\be \label{2.6}
<\Omega,\; W(0,F)\;\Omega>_{{\cal H}}=\delta_{F,0}
\ee
where $\delta_{F,0}$ is the Kronecker $\delta$. Hence the algebraic structure 
of $c(n)$ leads in a few lines to a representation of the Narnhofer-Thirring
type \cite{27}. The Hilbert space is the completion of the span $\cal D$
of the $W(0,F)\Omega$. The Weyl elements $W(f,0)$ act continuously, in fact
diagonally, on those 
by multiplication by $\exp(-i\;<F,f>)$ while the $W(0,G)$ act discontinuosly
by shifting $F$ to $F+G$. In fact, by Stone's theorem, also $<E,f>$ is 
well defined and acts by multiplication by $-<F,f>$.
The Hilbert space is not separable whenever the 
set of admitted $F$ is not countable.  

The representation theory of (\ref{2.2}) is really as simple as that
as soon as we agree that $\Omega$ be annihilated by $E$, with no conditions 
on $f,F$ except that $<F,f>$ should be a well defined number. This still 
allows $f,F$ to be distributions on $\sigma$ subject to the condition that 
their singularity structure is weighted in such a way that $<F,f>$ be well 
defined. We now show that the algebraic form of the Hamiltonian 
constraints {\it uniquely dictates 1. the smearing dimensions of 
$A,E$, 2. the density weight of the Hamiltonian constraint and 3. in any 
dimension.} For the formulation of quantum gravity in connection variables 
for all dimensions see \cite{28}.
\begin{Theorem} \label{th2.1} ~\\
Consider quantum gravity in $D+1$ spacetime dimensions in a 
representation of Narnhofer-Thirring type as above. Suppose that there is 
at least a cosmological constant in addition to the vacuum contribution.
Then the smearing dimension of $A,E$ are $1,D-1$ resepectively and 
the density weight of the Hamiltonian constraint must be unity.
\end{Theorem}
\begin{proof}:\\
In D spatial dimensions we have with $E^a_j=\sqrt{\det(q)}\; e^a_j$ that 
$|\det(E)|=\det(q)^{(D-1)/2}$. Therefore the cosmological term of the 
density weight $w$ Hamiltonian constraint is given by 
\be \label{2.7}
\Lambda\; \int_\sigma\; d^Dx\; n(x)\;
|\det(E)|^{w/(D-1)}(x)
\ee
while the kinetic term of the vacuum contribution contains the term
\be \label{2.7a}
\int_\sigma\; d^Dx\; n(x)\; {\rm Tr}([A\cdot E]\;[A\cdot E])\; 
|\det(E)|^{(w-2)/(D-1)}
\ee
where $A\cdot E$ denotes some contraction of tensorial and Lie algebra 
indices which depend on $D$ and which is not important for the proof.

Now $E$ is diagonal on the $W[(0,F)]$ with eigenvalue $F$. Thus in 
(\ref{2.7a}) the action of the operator produces in particular 
the term ${\rm Tr}([A\cdot F]\;[A\cdot F])$ which we have to quantise 
e.g. in terms of $\sin(<F_\Box,A>)$ in a Riemann sum regularisation 
of the integral by summing over cells $\Box$ of coordinate volume $\epsilon$ 
and where $F_\Box$ is the restriction of $F$ to such a cell. This 
means that for such a cell the contribution of (\ref{2.7a}) reads schematically
\be \label{2.7b}
n(p)\; {\rm Tr}([\sin(<F_\Box,A>)]\;[\sin(<F_\Box,A>)])\; 
\epsilon^{-D}\; |\det(F)|^{(w-2)/(D-1)}
\ee
where $p$ is the center of the cell. 

Thus if $A$ 
is smeared in $k=0,1,.., D$ dimensional submanifolds, 
then $F$ contains $D-k$ $\delta$ distributions and correspondingly 
$\det(F)$ will contain $D(D-k)$ $\delta$ distributions, isotropically 
wrt  
direction dependence, that is, it will contain $D-k$ $\delta$ distributions 
in each coordinate direction. From (\ref{2.7b}) we see that $k=D$ is 
not possible because then the limit $\epsilon\to 0$ would be singular,
rather the $D$ factors of $\epsilon$ in the denominator must turn into 
something finite upon replacing the denominator by 
an integral 
\be \label{2.7c}
\int_\Box\; d^Dx\;
|\det(F)|^{(2-w)/(D-1)}(x)  
\ee
and the integrand must have the singularity structure 
of the $\delta$ distribution in $D$ dimensions in order to remain finite as 
$\epsilon\to 0$.

Likewise the action of (\ref{2.7}) also just replaces $E$ by $F$. In order 
that this term be also finite we obviously must have equal powers of the 
determinant
\be \label{2.7d}
\frac{w}{D-1}=\frac{2-w}{D-1}\;\; \Rightarrow\;\; w=1
\ee
and in order to produce the $\delta$ distribution in $D$ dimensions 
we must have 
\be \label{2.7e}
(D-k)\frac{w}{D-1}=1 \;\; \Rightarrow \;\; k=1
\ee
\end{proof}
~\\   
It follows that for $D=3$ the functions 
$F,f$ smear $A,E$ effectively in 1 and 2 dimensions 
respectively, i.e. they are concentrated on 1 and 2 dimensional submanifolds 
respectively (that is curves $c$ and surfaces $S$ respectively). If one wants 
in addition that $W(0,F), W(f,0)$ transform covariantly under SU(2) gauge 
transformations then one considers instead of $W(F,0)$ holonomies 
$H(c)$ of $A$ along curves $c$ and instead of $W(0,f)$ fluxes 
$\Phi_f(S)$ of ${\rm Tr}(Ef)$ through surfaces $S$ where $f$ is an su(2) 
valued function on $S$. In this way one arrives naturally at the 
holonomy flux algebra and its LQG representation 
\cite{11} using mostly dynamical input.

The discussion reveals that choosing density weight $w\not=1$ while
keeping smearing dimensions $1,2$ for $F,f$ 
as is done in \cite{21} makes both terms 
(\ref{2.7}) and (\ref{2.7a}) formally diverge for $w>1$ and trivial for 
$w<1$ in the limit $\epsilon\to 0$. 
For the cosmological term that can be rigorously shown e.g. in the
weak operator topology of    
the LQG representation because the operators $\Phi_f(S)$ 
and not only their exponentials exist. For (\ref{2.7a}) this argument cannot
be made because while one can replace (\ref{2.7a}) by a Riemann 
sum approximant of the structural form 
\be \label{2.9}
\epsilon^{-3(w-1)}\; \sum_v \;n(v)\;
{\rm Tr}([H_\epsilon(v)-H_\epsilon(v)^{-1}]^2\; \Phi_\epsilon(v)^2)\; 
|\det(\Phi_\epsilon(v))|^{w/2-1}
\ee
where $H_\epsilon(v),\;\Phi_\epsilon(v)$ denote holonomies and fluxes 
respectively localised in cubes of coordinate volume $\epsilon^3$ and centre 
$v\in \sigma$, the limit $\epsilon \to 0$ does not exist say in the weak
operator topology because the holonomies are not weakly continuous in 
the LQG representation. In \cite{21} one picks the non-standard 
density weight $w=4/3$ which yields a prefactor $\epsilon^{-1}$ in front
of the sum in (\ref{2.9}). We will see the motivation for doing 
this in the subsequent sections.  

For the time being, we note that for $w=1$ a term like (\ref{2.9}) does 
converge in the following non-standard operator topology: Let $l\in L$ 
a diffeomorphism invariant distribution on $\cal D$ the linear 
span of the $W(0,F)\Omega$ or the Pol$(\{H\})\Omega$ where Pol
denotes polynomials of holonomies. Let $\psi\in {\cal D}$. Then 
the operator $O_\epsilon(n)$ corresponding to (\ref{2.9}) can be 
defined and evaluated on $\psi\in {\cal D}$ \cite{15}. Then $O_\epsilon(n)$ 
converges to the operator $O(n)$ densely defined on $\cal D$ if 
\be \label{2.11}
\lim_{\epsilon\to 0}\; l[(O_{\epsilon}(n)-O(n))\; \psi]=0 \;\; 
\forall\; \psi\in {\cal D},
\;\; l\in L
\ee
One finds that the limit is trivial and one may pick any fixed  
$\epsilon=\epsilon_0$ and set $O(n)=O_{\epsilon_0}(n)$. The reason why 
this works is because $w=1$: This makes the whole construction 
diffeomorphism covariant and changing $\epsilon$ can be absorbed into 
a diffeomophism to which $l$ is insensitive. This does not work 
for any other density weight.   

We close this section by mentioning that a similar argument as above 
applies for all the matter content of the standard model 
and uniquely fixes the smearing dimensions whenever one of the 
members of the canonical pair annihilates the vacuum $\Omega$, see 
\cite{12,15} for details.  

\section{Non-degenerate states and density weight}
\label{s3}

To understand the apparent tension between the natural density one 
weight of $C(x)$ and proper representation of of the hypersurface 
deformation algebra $\mathfrak{h}$ in the LQG representation and its 
relation to the non-degenracy condition, it is necessary to go into 
more details. The rest of the paper considers the case $D=3$ only. \\
\\
As outlined in the previous section, it is well motivated to work 
in a representation in which the vacuum is annihilated by the 
2d smeared operator corresponding to $E^a_j$. The corresponding
LQG Hilbert space is equipped with an ONB known as SNWF. These are 
certain polynomial functions of an arbitary set of holonomies along 
1d oriented curves called edges $e$ that intersect nowhere except in their 
endpoints   
called vertices $v$. The edges are labelled by half integral 
spin quantum numbers $j_e$ while the vertices are labelled by intertwiners 
$\iota_v$ between the corresponding irreducible representations meeting 
at $v$. Hence a SNWF $T_s$ is labelled by a spin network $s=(\gamma,j,\iota)$
where $\gamma$ is a finite, oriented graph and $j,\iota$ are the collection of 
$j_e,\iota_v$.
  
The classical function $C$ on the phase space coordinatised by $A,E$ 
is polynomial in $A$ but not 
a polynomial in $E$. Rather it also depends on integer inverse powers of 
$|\det(E)|^{1/2}$. It is a non-trivial result in LQG that the integral 
of $|\det(E)|^{1/2}$ over 3d submanifolds $R$ is a well defined, in fact 
essentially s.a. operator 
known as the volume operator $V(R)$ whose dense domain is 
given by the span $\cal D$ of the SNWF. We use here as in \cite{15} the 
version of $V(R)$ due to Ashtekar and Lewandowski because only with 
this version the operator $C(n)$ defined below annihilates vertices with 
co-planar tangents of adjacent edges and only this operator passes the 
triad test \cite{14}, i.e. it implements the classical non-polynomial identity
\be \label{3.1}
E^a_j(x)=\frac{1}{2}\epsilon^{abc}\epsilon_{jkl}\;
{\rm sgn}(\det(\{V(R),A(x)\}) 
\{V(R),A_a^k(x)\} 
\{V(R),A_a^k(x)\} 
\ee
for any $x\in R$. If co-planar vertices are not annihilated then 
$C(n)$ is not densely defined on $\cal D$.

To define inverse powers of $|\det(E)|^{1/2}$ one now uses classical 
approximate identities of the form \cite{15} 
\be \label{3.2}
q^3\; V(R)^{1-3(1-q)}\approx \int_R d^3x\;  \det(\{V(R)^q,A(x)\})  
\ee
for $0<q<2/3$ to obtain $V(R)^{-p},\; 0<p<2$ for small regions $R$ 
of coordinate volume $\epsilon^3$  and approximates the 
integral in (\ref{3.2}) in terms of Poisson brackets with three holonomies.
This rewrites negative powers of $V(R)$ in terms of positive powers 
and commutators with holonomies and the latter operators are well
defined in the LQG representation. 

To tame the higly non-polynomial expressions that appear in $C$ in its 
quantisation one proceeds as follows \cite{15} (for illustrational 
purposes we consider here 
only the Euclidan part of the Lorentzian constraint and only the 
geometry contribution, Lorentzian part, cosmological constant and 
matter part can be treated by similar methods \cite{15}):
We partition the manifold into cells $\Box_p$ of coordinate volume 
$\epsilon^3$ and centre $p$ and approximates the integral involved 
in $C(n)$ as the Riemann sum $\sum_p\;n(p)\;C_(\Box_p)$ where $C_(\Box_p)$
is the integral of $C(x)$ over $\Box_p$. Then one approximates 
the ``electrical part'' of integrand of $C(\Box_p)$ in terms of 
quantities of the form (\ref{3.2}) with $R=\Box_p$ and 
$A$ replaced by holonomies along edges of coordinate length 
$\epsilon$ and the ``magnetic part'' in terms 
of holonomies along loops enclosing a surface of coordinate area 
$\epsilon^2$. When one now acts on a SNWF $T_\gamma$ over a 
graph $\gamma$, due to the properties 
of the volume operator mentioned, for sufficiently 
small $\epsilon$ one finds $C(\Box_p)\;T_\gamma=0$ unless 
$\Box_p$ contains a vertex $v\in V$ where $V$ is the vertex set 
of $\gamma$. 
This is because $V(R)$ acts non-trivially only at at least 
4-valent non-co-planar gauge invariant vertices $v$ or at least 3-valent 
non-coplanar non gauge invariant vertices $v$ and only if $v\in R$.   

One further approximates $p=v$ for such $p$ with $v\in \Box_p$ and  
chooses the edges involved in the electric part to be beginning
segements $s(e)$ of edges $e\in E$ adjacent to $v$ and the loops 
involved in the magnetic part along pairs $s^\epsilon(e),s^\epsilon(e')$ with 
$e\not=e'$ adjacent to $v$ connected by an ``arc'' $a^\epsilon_{v,e,e'}$
so that a loop $\alpha^\epsilon_{v,e,e'}=
s^\epsilon(e)\circ a^\epsilon_{v,e,e'}\circ s^\epsilon(e')^{-1}$ 
results with $\alpha^\epsilon_{v,e',e}=[\alpha^\epsilon_{v,e,e'}]^{-1}$ 
which does 
not intersect $\gamma$ except at interior points of $e,e'$. Then the 
final regulated operator acts on SNWF as  
\be \label{3.3}
C^\epsilon(n)\; T_\gamma=\sum_v\; n(v)\;\sum_{e\cap e'=v}
C^\epsilon_{v,e,e'} \; T_\gamma
\ee
In the non-standard topology mentioned above, we may choose for 
each $\gamma$ a sufficiently small $\epsilon(\gamma)$ meeting the conditions 
used in the derivation of (\ref{3.3}) and then the limit $\epsilon \to
0$ becomes trivial and amounts to define the limit operator by
\be \label{3.4a}
C(n)\; T_\gamma=\sum_v\; n(v)\;\sum_{e\cap e'=v}
C_{\gamma,v,e,e'}\;\; T_\gamma,\;\;
C_{\gamma,v,e,e'}:=C^{\epsilon(\gamma)}_{v,e,e'} 
\ee
Since $\epsilon(\gamma)$ is quite arbitrary, (\ref{3.4}) suffers from 
that arbitrariness, however, when looking for solutions of the constraints
$C(n),U(\varphi)-1$, that is, distributions which vanish 
on $C(n)\; T_\gamma,\;[U(\varphi)-1] \; T_\gamma$ for all $n,\varphi,
T_\gamma$ most of these ambiguities are washed away because (\ref{3.4})
is diffeomorphism covariant: Because the density weight is 
unity, by construction opertators 
resultiong from different choices of 
$\epsilon(\gamma)$ are related by a diffeomorphism so that the joint kernel 
is the same. Differences can only occur from different diffeomorphism 
equivalence classes of the loops $\alpha_{v,e,e'}$.

We can now proceed to compute the commutator of the $[C(m),C(n)]$. One finds 
\be \label{3.4}
\frac{1}{2}\;\sum_{v\not=v'}\; [m(v)n(v')-n(v) m(v')] \;
\sum_{e\cap e'=v}\;
\sum_{f\cap f'=v'}\;
[1-U(\varphi_{v,e,e';v',f,f'})]\;
C_{\gamma_{v',f,f'},v,e,e'},\; C_{\gamma,v',f,f'}\;
T_\gamma
\ee
where, again due to the properties of the volume opertor, 
the sums only involve vertices and edges of $\gamma$ and not 
the new edges and vertices generated by the arcs. The second graph is 
$\gamma_{v',f,f'}=\gamma\cup a_{v',f,f'}$ and the diffeomorphism 
displayed preserves $\gamma$ but 
takes care of the fact that the second action depends on whether 
one first acted at $v$ or $v'$. It was also used that the two 
$C$ expressions displayed commute for $v\not=v'$.     
   
On the other hand, the direct 
quantisation of the classical $K(m,n)=\{C(m),C(n)\}$ 
leads to an operator with 
similar action and properties. It is given by \cite{17}
\be \label{3.5}
K(m,n)=\sum_v\; \sum_{e\cap e'=v}\; [m(v)n(v+s(e))-n(v) m(v+s(e))] 
[U(\varphi_{\gamma,v,e'}-1] Q^{ee'}\; T_\gamma
\ee
where $s(e)=s^{\epsilon(\gamma)}(e)$
is again the segment of $e$ beginning at $v$ and 
$Q^{ee'}$ is a 
geometrical operator which quantises $q^{-1}$ by the methods above
independent of the choice $\epsilon(\gamma)$ and 
$\varphi_{\gamma,v,e}=\varphi^{\epsilon(\gamma)}_{v,e'}$ is 
a diffeomorphism that acts non-trivially only in the $\epsilon$ cube. 
Also this operator contributes only at the vertices of $\gamma$.
One could be 
content with this if the diffeomorphisms involved in (\ref{3.4}) and 
(\ref{3.5}) would depend on the same data and if $Q^{e,e'}$ could be 
related to the two $C$ expressions in (\ref{3.4}). This is, however, not 
possible because 
the vertices $v,v'$ belong to $\gamma$ while the vertex $v+s(e)$ does not
and because the two $C$ expressions mutually commute and thus are 
not able to produce an operator resembling $Q^{ee'}$ via a commutator.  
In order to improve this the second action of $C$ would need to involve also
the new vertices that the first action creates via the arcs 
but this leads to a 
mathematical anomaly. This is in more detail what we meant by closure 
with the wrong structure functions, while both (\ref{3.4}) and (\ref{3.5})
annihilate the space $L$ of diffeomorphism invariant distributions.\\
\\
We now switch to the framework of \cite{19} where one tries to take the 
limit $\epsilon\to 0$ using dual action of $C^\epsilon(n)$ on a subspace 
$L'$ of distributions on $\cal D$. One picks the $l\in L'$ of the 
general form 
\be \label{3.6}
l=\sum_{s'}\; l(s')\; <T_{s'},.>_{{\cal H}}
\ee
where the sum is over all SNW labels $s'$ with certain coefficients $l(s')$. 
Then
\be  \label{3.7}
l[C^\epsilon(n) T_\gamma]=\sum_v \; n(v)\;\sum_{e\cap e'=v}\;
\sum_{s'\in S^\epsilon_{s,\gamma,v,e,e'}} l(s')
<T_{s'},C^\epsilon_{\gamma,v,e,e'}\; T_s>
\ee
where $S^\epsilon_{s,\gamma,v,e,e'}$ is the set of all SNW into which 
$C^\epsilon_{\gamma,v,e,e'}\; T_s$ decomposes which are finitely many. 

The observation is now the same that makes the non standard operator 
topology limit work: The matrix elements 
$<T_{s'},C^\epsilon_{\gamma,v,e,e'}\; T_s>$ do not depend on $\epsilon$ 
because the inner product is diffeomorphism invariant (the diffeomorphisms 
act unitarily). Therefore, the only $\epsilon$ dependence in 
(\ref{3.7}) rests in the coefficient $l(s')$. If we choose $L'$ to consist
of those $l$ such that $l(s')$ is a  
continuous functional of $s'=(\gamma',j',\iota')$
with respect to its graph entry $\gamma'$ (e.g. a function 
of edge length with respect to a background metric) then the limit
$\epsilon\to 0$ can be carried out. Setting
$l_n(T_s):=\lim_\epsilon l(C^\epsilon(n) T_s)$ then a dual operator 
$C'(n)\;l:=l_n$ is defined. It is now easy to see that 
$[C'(m),C'(n)]=0$ because the commutator depends on different vertices 
and $L'$ consists of jointly continuous functions. One can even 
generalise this to more general actions of $C(n)$ which also 
involve the new vertices created \cite{19}. Note that $L'$ contains 
$L$ (for which $l(s)$ is the constant function for $s$ in the same 
diffeomorphism equivalence class).

For the same choice of $L'$ we can construct the dual action  of 
(\ref{3.5}) in the limit $\epsilon\to 0$. The result trivially 
vanishes if the lapse functions $m,n$ are continuous by the same argument. 
Thus we get consistently $[C'(m),C'(n)]=K'(m,n)=0$.\\
\\
This latter observation motivates the non-standard density weight: Similar to
\cite{21} consider $w=\frac{7}{6}$: It leads to an additional $\epsilon^{-1}$ 
factor in front of an expression the limit of whose action on $L'$ is 
finite, see (\ref{2.9}). Classically we have for $w=7/6$ that 
$C_w=C\;[\det(q)]^{1/12}$ and by the $\mathfrak{h}$ relations that 
$\{C_w(m),C_w(n)\}=-D(q^{-1}(m dn-n dm)[\det(q)]^{1/6})=:K_w(m,n)$ so that
$K_w(m,n)$ can 
be approximated by an expression of the symbolic form
\be \label{3.8}
\sum_v \;
{\rm Tr}([H_\epsilon(v)-H_\epsilon(v)^{-1}]^2\; \Phi_\epsilon(v)^3)\; 
\frac{m(v)n(v+\epsilon)-n(v)m(v+\epsilon)}{\epsilon}
|\det(\Phi_\epsilon(v))|^{-5/6}
\ee
which can be quantised using (\ref{3.2}) and above Riemann sum techniques
on $\cal D$. Now to define $K_w'(m,n)$ on some $L'$ will involve the 
limit
\be \label{3.8}
\lim_{\epsilon \to 0}\; 
\frac{m(v)n(v+s^\epsilon(e))-n(v)m(v+s^\epsilon(e))}{\epsilon} 
=\dot{e}^a(0)\;[m\;\partial_a n-n\;\partial_a m](v)
\ee
where $t\mapsto e(t),\;e(0)=v$ denotes the parametrisation of $e$
and thus produces exactly the combination of lapse functions that one 
has in the classical theory. The challenge then consists in 
quantising $C_w(m)$ for $w=7/6$ in such a way that 
$[C'_w(m),C'_w(n)]=i\;K'_w(m,n)$ on suitable $L'$ 
which involves letting $C_w(m)$ 
act on the new vertices it creates and to interpret  
combinations such as $\tilde{u}^a_{n,j}:=n E^a_j\;|\det(E)|^{-5/12}$ as 
three phase space dependent shift vector fields called electric 
shifts which motivates to let $C_w$ act on SNWF in a similar way as 
a linear combination of finite diffeomorphism operators would do, 
but with operator valued coefficients. 
The associated deformations caused by these diffeomorphisms have to be 
chosen carefully in order that the $\epsilon^{-1}$ factor in (\ref{2.9}) 
causes no singularity. The choice of $L'$ needs some form of analytic 
structure which is fed in by using dependence of $l(s')$ on a background 
metric.\\
\\
There is a price to pay whatever choice of $w$ one makes: For 
$w=1$ 
the norm of (\ref{3.3}) in the LQG Hilbert space is {\it independent
of $\epsilon$} and finite (due to diffeomorphism covariance of the construction
and this in fact motivates the non-standard topology) while  
the norm of (\ref{3.4}) in the LQG Hilbert space {\it converges to zero} for 
continuous $M,N$ as $\epsilon\to 0$. For $w=7/6$ it is opposite:
The norm in the LQG Hilbert space 
of the analog of (\ref{3.3}) {\it diverges} as $\epsilon\to 0$ while 
the norm of the analog of (\ref{3.4}) in the LQG Hilbert space 
{\it converges to a finite limit} which correctly depends on the 
wanted combinations of derivatives $M\partial N-N\partial M$.\\
\\
\\
It is now time to unveil the reason for why with standard density weight
$w=1$ the dual algebra $[C'(m),C'(n)]=0$ {\it must be Abelian} for the 
choice of $L'$ made in \cite{19} and its generalisations. In other 
words {\it its Abelian nature is physically correct}. To see 
this we note that 
$L'$ is a space of distributions over ${\cal D}={\cal D}_{{\rm SNWF}}$,
the {\it finite} linear span of SNWF. This is a dense and invariant 
domain for $C(n)$ because $C(n)$ acts only at the vertices of a graph
and the graphs involved in SNWF are finite, the number of SNWF involved 
in $\psi\in {\cal D}$ is finite. Therefore there is a substantial 
difference between the quantum state $C^\epsilon(n)\psi,\;\psi\in {\cal D}$ 
and the 
classical expression $C^\epsilon(n)$ (the Riemann sum regularisation 
of $C(n)$ sketched above): 
The former is a sum over a {\it finite} (say $N$) number of cells where 
$N$ is the number of vertices involved in $\psi$ while $C^\epsilon(n)$ 
is an {\it infinite} sum (for non compact $\sigma$; for compact $\sigma$ 
the number of cells still grows indefinitely as $\epsilon\to 0$). We
pick $\psi=T_\gamma$ and denote the situation symbolically as 
\be \label{3.9}
C^\epsilon(n)\; T_\gamma=\sum_{v\in V}\; n(v)\; C(\Box^\epsilon_v)\;
T_\gamma,\;\;
C^\epsilon(n)=\sum_{p\in P}\; n(p)\; C(\Box^\epsilon_p)
\ee
where $P$ denotes the set of centre points used in the  
partition into cubes $\Box^\epsilon_p$ of $\sigma$.
Computing commutators and Poisson brackets respectively 
yields again symbolically 
\ba \label{3.10}
[C^\epsilon(m),C^\epsilon(n)]            
\; T_\gamma &=& \frac{1}{2} \; 
\sum_{v,v'\in V}\; 
[m(v) n(v')-m(v') n(v)]\;
[C(\Box^\epsilon_v),C(\Box^\epsilon_{v'})]\; T_\gamma
\nonumber\\
\{C^\epsilon(m),C^\epsilon(n)\}
&=& \frac{1}{2}
\sum_{p,p'\in P}\; 
[m(v) n(v')-m(v') n(v)]\;
\{C(\Box^\epsilon_p),C(\Box^\epsilon_{p'})\}
\ea
Let us compute the classical Poisson brackets explicitly
for the explicit form used in \cite{15} 
\be \label{3.11}
C(\Box^\epsilon_p):=\sum_{a,b,c=1}^3 \;\epsilon^{abc}\; 
{\rm Tr}([H_{ab}(p)-H_{ba}(p)] H_c(p)\; \{V(\Box^\epsilon_p), 
H_c(p)^{-1}\})
\ee
where $H_a(p)$ is the holonomy from $p$ in direction $a$ by one unit 
of $\epsilon$ and 
$H_{ab}(p)=H_a(p)\;H_b(p+\delta_a)
H_a(p+\delta_b)^{-1}\;H_b(p)^{-1}$ is a plaquette holonomy while 
$V(\Box^\epsilon_p)=|\det(\delta\Phi(p)|^{1/2}$ where
$[\delta \Phi]^a_j(p)=\Phi^a_j(p)-\Phi^a_j(p-\delta_a)$ 
and
$\Phi^a_j(p)$ is the gauge covariant flux based at $p$ \cite{29}
through the boundary face of 
the two cubes
with centres $p,p+\delta_a$ with co-normal in $a$ direction 
and in direction $j$ wrt an ONB basis $\tau_j$ of su(2) wrt the trace metric. 
Since by construction 
\be \label{3.12}
\{\Phi^a_j(p),H_b(p')\}=\delta_{p,p'}\; \delta^a_b [\tau_j H_a(p)]
\ee
one finds to order $\epsilon^3$
\be \label{3.13}
\{C(\Box^\epsilon_p),C(\Box^\epsilon_{p'})\}=-
\sum_a\;\sum_{\sigma=\pm 1}\;\sigma\; \delta_{p',p+\sigma\delta_a}\;\; 
D^a_\epsilon(p),\;\;
D^a_\epsilon(p)=\sum_{b,c}\;
{\rm Tr}(H_{bc}(p)\; \Phi^c(p))\; 
\frac{{\rm Tr}(\Phi^a(p)\Phi^b(p))}{V(\Box^\epsilon_p)^2}
\ee
so that to leading order in $\epsilon$ 
\be \label{3.14}
\{C^\epsilon(m),C^\epsilon(n)\}
=-\frac{1}{2}\sum_{p\in P}\;\sum_a 
[m\;(\partial^\epsilon_a\; n)
-n\;[\partial^\epsilon_a\; m)](p)\; D^a_\epsilon(p)
\ee
with the lattice derivative $(\partial^\epsilon_a m)(p)=
m(p+\delta_a)-m((p-\delta_a)$. If we take the limit $\epsilon\to 
0$ of (\ref{3.14}) then we recover precisely $-D[q^{-1}(m\; dn-n\;dm]$
because 1. $D^a_\epsilon(p)$ is of order $\epsilon^2$, 2.
$m\partial^\epsilon_a n-n\partial^\epsilon_a m$ is of order $\epsilon$
and 3. and mosty importantly for the main argument of this work 
{\it the number of terms in the sum grows as $\epsilon^{-3}$}.
 In order that this holds, we need {\it all the contributions $p\in P$}
in order that the Riemann sum $\sum_{p\in P}\; \epsilon^3$ turns 
into the integral $\int_\sigma \;d^3x$. We also need that 
$V(\Box^\epsilon_p)>0$ for all $p,\epsilon$, i.e. that the classical 
metric is regular.

Let us now mirror this with the quantum computation in
(\ref{3.10}). Under the assumption 
that the commutator between the $v,v'$ contributions does not vanish 
identically we expect it to be non-vanishing at most when 
$v,v'$ are next neighbour vertices in the graph $\gamma$ since 
the quantum operator is constructed from local expressions. Let ${\cal N}_v$
be the set of next neighbour vertices $v'$ of $v\in V$. Then we obtain 
\be \label{3.15}
 [C^\epsilon(m),C^\epsilon(n)]            
\; T_\gamma=\frac{1}{2} \; 
=\sum_{v\in V}\;\sum_{v'\in {\cal N}_v} 
[m(v) n(v')-m(v') n(v)]\;
[C(\Box^\epsilon_v),C(\Box^\epsilon_{v'})]\; T_\gamma
\ee
Now even if in the best case the commutator left in (\ref{3.15}) is 
turned into a linear combination of diffeomorphism operators, there 
is no chance to match with (\ref{3.14}) because the number of 
terms involved is finite. Therefore, as for generic graphs $\gamma$
the next neighbours of $v$ are all away from $v$ by far more than $\epsilon$,
for such $\gamma$ (\ref{3.15}) trivially vanishes (perhaps modulo 
a diffeomorphism but that diffeomorphism has nothing to do with the 
diffeomorphism involved in $K(m,n)$). As this is automatically the case 
for sufficiently small $\epsilon$, there is no chance to match (\ref{3.14})
and (\ref{3.15}).    

In other words, even if in a semiclassical limit 
we have that the commutator in (\ref{3.15}) turns into the Poisson bracket 
(\ref{3.13}), that contribution is of order $\epsilon^2$ while 
$m(v) n(v')-m(v') n(v)$ is of order $\epsilon$ and we need an
order of $1/\epsilon^3$ terms to make the semiclassical limit non vanishing,
but there are only finitely many, namely $N$ of them.

Yet in other words, even if one would get an equality of the 
form $[C^\epsilon(m),C^\epsilon(n)]\;T)\gamma=K^\epsilon(m,n) T_\gamma$
we would find $K'(m,n)\equiv 0$ on any suitable $L'$ 
as soon as $K^\epsilon(m,n) T_\gamma$ is of the form
$\sum_{v,v'}\; [m(v) n(v')-m(v') n(v)]\; K_{\gamma,v,v'} T_\gamma$
with $K_{\gamma,v,v'}$ having diffeomorphism invariant matrix elements 
between a {\it finite number of} SNWF, with $m,n$ continuous
and with $l\in L'$ having continuous coefficients $l(T_s)$. \\
\\
\\
The discussion reveals that if one wants to avoid this triviality
and if one does not want to change the density weight, which as we showed in 
section \ref{s2} is problematic, then one must avoid that the sum over 
$P$ collapses to a sum over $V$ in (\ref{3.9}). More precisely, several 
conditions must be met at the same time in order that the quantum 
computation (\ref{3.15}) comes as close as possible to the classical 
compuation (\ref{3.14}). In the classical Riemann sum computation (\ref{3.14}) 
three things are happening simultaneously and are matched to 
each other: First a discretisation of space by cells, second a discretisation 
of the phase space labelled by those cells and third a discretisation of 
cell constraints by functions of the cell variables. In the quantum 
computation (\ref{3.15}) these three steps are also applied to the 
constraint operator {\it but not to the quantum state}. The quantum 
state still lives in the continuum and it is not subject to 
discretisation. Thus it is defined by a continuum of configuration 
quantum degrees of freedom (in the connection representation and 
spin degrees of freedom in the Fourier transformed representation).
Then two effects bring the classical and quantum computation drastically out 
of balance: first, due to the fact that SNWF are highly degenerate, 
the quantum constraint is ``blind'' for almost all of the cells, namely 
those that do not contain a vertex. Second, even for the cells that 
contain a vertex, the set of degrees of freedom that are changed on the quantum 
state by the action of the discretised constraint 
contain new ones with respect 
to which it was not already excited. In the literature this is referred 
to as ``graph changing action''. The fluctuations 
of these new excitations are therefore not 
controlled by the state one acts upon and thus avoid e.g. application of 
coherent state techniques.\\      
\\
To avoid the first effect, the state one acts upon should be non-degenerate.
To avoid the second effect, the discretised constraint should act on 
an invariant subspace of states.\\
\\ 
A first proposal that meets these 
two conditions is the the algebraic quantum gravity (AQG) programme 
\cite{32}. There one works with a single, fundamental, infinite abstract 
graph that can be embedded arbitrarily densely (i.e. with arbitrarily small
but finite spatial resolution) everywhere on $\sigma$ and the constraints,
which are considered as regulator free, preserve that fundamental 
abstract graph, but not all its subgraphs. Thus AQG is like a lattice 
gauge theory with the difference that the quantum state one acts upon 
decides about how densely the abstract graph is embedded. Since 
the abstract constraints of AQG do not close under commutators, the 
AQG framework was embedded into the master constraint programme 
\cite{10} which replaces all spatial diffeomorphism and Hamiltonian
constraints by a single one so that 
anomalies are of no immediate concern.  

A second proposal that meets these conditions is the Hamiltonian 
renormalisation programme \cite{18}. Here one still works with 
a Hilbert space $\cal H$ of concrete (embedded) 
states, however, the states in $\cal H$ are projected
into subspaces ${\cal H}_\epsilon = {\cal P}_\epsilon {\cal H}$ in a 
controlled way. The control consists in a partial order on the set $\cal E$ 
of labels $\epsilon$ with respect to which it is directed.
The $P_\epsilon$ arise as $P_\epsilon= J_\epsilon\;
J_\epsilon^\dagger$ where $J_\epsilon:\; {\cal H}_\epsilon \to {\cal H}$
is an isometric injection $J_\epsilon^\dagger J_\epsilon=1_{{\cal H}_\epsilon}$
which are fixed points of a renormalisation flow which grants that 
$\cal H$ is the inductive limit of the ${\cal H}_\epsilon$ which 
may or may not coincide with the LQG Hilbert space. The 
discretised constraints are also subject to renormalisation and yield 
a consistent family $C_\epsilon(n)$ at the fixed point (if it exists)
in the sense that 
\be \label{3.16}   
C_\epsilon(n)=J_{\epsilon\epsilon'}^\dagger 
C_{\epsilon'}(n)\;J_{\epsilon\epsilon'}
\ee
with $J_{\epsilon\epsilon'}=J_{\epsilon'}^\dagger J_\epsilon$
for all $\epsilon<\epsilon'$ which thus grants existence of 
a quadratic form $C(n)$ such that
\be \label{3.16}   
C_\epsilon(n)=J_{\epsilon}^\dagger 
C(n)\; J_{\epsilon}
\ee
The $C_\epsilon(n)$ {\it must not close} under commutators, even if the 
$C(n)$ do. An example where the validity of this procedure has been 
recently demonstrated is parametrised field theory which among other things 
also displays a non-trivial realisation of the hypersurface deformation 
algebra \cite{33}, in fact for density weight two rather than one. 
On the other hand,
in the corresponding
Hilbert space representation (a Fock representation) 
the metric operator does not annihilate the vacuum so that all Fock
states are non-degenerate. We therefore revisit PFT with density weight 
one and with degenerate vacuum in section \ref{s5} and show that 
nevertheless one can get the algebra to close.

A third proposal is to consider the representation \cite{30} 
of the holonomy flux algebra 
different from the LQG representation. It modifies it  
by a condensate $<\Omega^0,\; \Phi_f(S)\;\Omega^0>=\Phi^0_f(S)$ 
where $\Phi^0$ is a classical electric field. 
Choosing $\Phi^0$ to correspond to a non-degenerate 
metric, then $\Omega^0$ is a non-degenerate vacuum 
state in the sense of the next section. To see whether in this 
representation we can hope to make progress wrt the representation of 
$\mathfrak{h}$ we consider formally $C(n)\;w[F]\Omega_0$ which formally 
can be written
\be \label{3.17}
\int\;d^3x\; n\; B^a_j\; [e(\Phi_0+F)]_a^j\;w[F]\;\Omega_0,\;
e(G)_a^j=\frac{1}{2}\epsilon^{jkl} \epsilon_{abc} 
\frac{G^b_k\;G^c_l}{|\det(G)|^{1/2}} 
\ee
Note that $F$ is a distribution with the singularity structure of 
a $\delta$ distribution in 2 dimensions while $\Phi_0$ is smooth. 
Also the magnetic field itself is ill-defined.
To regularise (\ref{3.17}) 
we use a Rieman sum approximation of the 
integral by a sum over $\epsilon$ sized cells $\Box$ with centre 
$p_\Box$
and $\epsilon^2 B^a_j$ replaced by 
$B^a_j(\Box):={\rm Tr}(H(\alpha^a_\Box))\tau_j$ where 
$\alpha^a_\Box$ is an appropriate loop located in $\Box$ in the coordinate 
plane transversal to the $a$ direction. In order that the 
denominator $|\det(\Phi_0+F)|^{1/2}$ turns into something finite, we 
integrate
it over $\Box$ which in the limit $\epsilon\to 0$ makes the $\Phi_0$ 
dependence disappear from the denominator. The numerator then 
depends schematically on the term $\epsilon^4 (\Phi_0+F)^2$ where $\epsilon^4$
comes from the left over $\epsilon$ of the measure factor $\epsilon^3$ not 
absorbed by $B$ and the fact that we have to multiply both 
numerator and denominator by $\epsilon^3$ when we integrate the denominator 
over $\Box$. Thus the $\Box$ contribution to the 
numerator contains the three terms, schematically
\be \label{3.18}    
B(\Box)\;[\Phi_0(\Box)^2+2 \Phi_0(\Box)(F(\Box)+F(\Box)^2]
\ee
where $\Phi_0(\Box),F(\Box)$ are integrals over faces of $\Box$.
Suppose we use some of the methods of \cite{15} to define 
$|\det(F(\Box)|^{-1/2}$ which thus lets only those $\Box$ contribute that 
contain a vertex of the graph. Then the contribution of the 
first two terms in (\ref{3.18}) to the norm of (\ref{3.17}) vanishes 
in the limit $\epsilon\to 0$ because the number 
of contributing $\Box$ is constant 
so that altogether nothing has changed as 
compared to the situation without condensate. To change something, as argued 
above, 
all $\Box$ must contribute, thus $F$ itself must be excited already 
everywhere so that the denominator is finite for every $\Box$. 
Thus it appears that $\Phi_0$ by itself is not sufficient 
to achieve a non-anomalous $\mathfrak{h}$. \\
\\
\\
These qualitative  arguments suggest that we need to consider $F$ that 
correspond to an everywhere
excitated 
quantum geometry. This should define a new dense domain different from 
the finite 
linear span of spin network functions. 
In the next section we investigate qualitatively 
how this might be achieved by combination of coherent state and renormalisation 
methods.

\section{Qualitative investigation of $\mathfrak{h}$ in Hamiltonian 
renormalisation of LQG}
\label{s5}

As mentioned before and reviewed in \cite{18} in Hamiltonian renormalisation
we construct a sequence labelled by $s$ of families labelled by $\epsilon$
of triples $({\cal H}^{(s)}_\epsilon, J^{(s)}_{\epsilon,\kappa(\epsilon)},
C^{(s)}_\epsilon(n))$ where $\epsilon':=\kappa(\epsilon)\le 
\epsilon$ is a fixed element 
wrt the partial order $\ge$ (dictating how many degrees of freedom of the 
finer theory labelled by $\epsilon'$ are integrated out to reach the 
coarser theory labelled by $\epsilon$) and the entries of the triple
are Hilbert spaces, isometric embeddings $J^{(s+1)}_{\epsilon\epsilon'}\; 
{\cal H}^{(s+1)}_\epsilon\to  {\cal H}^{(s)}_{\epsilon'}$ and constraints 
respectively. This isometry condition together with the prescription 
\be \label{5.1}
C^{(s+1)}_\epsilon(n):=
[J^{(s+1)}_{\epsilon\epsilon'}]^\dagger\;
C^{(s)}_{\epsilon'}\;
[J^{(s+1)}_{\epsilon\epsilon'}],\; \epsilon'=\kappa(\epsilon)
\ee
defines a renormalisation flow starting from an initial triple that uses 
a classical discretisation such as the Riemann sum approximations that
were used in section \ref{s3}. While the initial triple suffers from 
many ambiguities, the intuition collected from statistical physics examples   
gives rise to the hope that theories labelled by different ambiguity
parameters flow into the same fixed point family
${\cal H}_\epsilon, J_{\epsilon,\epsilon'},C_\epsilon(n))$ 
defining a continuum Hilbert space $\cal H$ as the inductive limit of the 
${\cal H}_\epsilon$ and continuum operators $C(n)$ such that there exist
isometric embeddings $J_\epsilon:\; {\cal H}_\epsilon\to {\cal H}$ with
\be \label{5.2}
J_\epsilon=J_\epsilon'\; J_{\epsilon\epsilon'},\;\;
C_\epsilon(n)=J_\epsilon^\dagger\; C(n)\; J_\epsilon
\ee
The $C_\epsilon(n)$ {\it must not close} even if the $C(n)$ {\it do close}
since 
\be \label{5.3}
[C_\epsilon(m),C_\epsilon(m)]=
J_\epsilon^\dagger\; [C(m),C(n)]\; J_\epsilon
-
J_\epsilon^\dagger\; \{
C(m)\;[1_{{\cal H}}-P_\epsilon]\;C(n)
-C(n)\;[1_{{\cal H}}-P_\epsilon]\;C(m)\}\; J_\epsilon
\ee
where $P_\epsilon=J_\epsilon J_\epsilon^\dagger$ is the projection 
of the continuum theory into a subspace isomorphic to the discretised 
theory at resolution $\epsilon$. The correction terms proportional to 
$1_{{\cal H}}-P_\epsilon$ are expected to converge to zero as 
$\epsilon\to 0$ e.g. in the weak operator topology on $\cal H$, that is,
given $\psi,\psi'\in {\cal H}$ the correction terms in (\ref{5.3}) are 
expected to vanish at fixed $\psi,\psi'$ when sandwiched between 
$P_\epsilon \psi, P_\epsilon \psi'$ provided that $1_{{\cal H}}-P_\epsilon$ 
itself converges to zero when sandwiched between $\psi,\psi'$. These 
expectations are met in PFT \cite{33}.

The strategy to check for anomaly freeness of $\mathfrak{h}$ 
in Hamiltonian renormalised LQG would therefore be as follows:\\
Step 1:\\
Start with initial families of Hilbert spaces
$({\cal H}^{(0)}_\epsilon:=L_2(d\nu^{(0)}_\epsilon,{\cal A}_\epsilon)$  
defined as square integrable functions with respect to some measure 
$\nu^{(0)}_\epsilon$ on a space of connections ${\cal A}_\epsilon$ as well
as with discretised constraints $C^{(0)}_\epsilon(n)$ and discretised
``would be'' commutators $K^{(0)}_\epsilon(m,n)$ thereof. \\
Step 2:\\
Construct the flow of these where isometry translates into a 
flow of measures $s\mapsto \nu^{(s)}_\epsilon$.\\
Step 3\\
Compute the corresponding fixed points and check $\mathfrak{h}$.\\
\\
Since the completion of step 3 will be very difficult in practice 
because the computation of the fixed point will require a large number 
of iterations of the renormalisation step, we consider the following 
algebra check after a finite number $s$ of iterations:\\
Step 3$_{s,N}$:\\ 
Consider the matrix elements of the combination
\be \label{5.4}
[C^{(s)}_\epsilon(m),C^{(s)}_\epsilon(n)]-
K^{(s)}_\epsilon(m,n)
\ee
between $P^{(s)}_{\epsilon,\epsilon'_N}\; \psi_{\epsilon'_N},\;
P^{(s)}_{\epsilon,\epsilon'_N}\; \psi'_{\epsilon'_N}$ for fixed
$\psi_{\epsilon'_N},\;\psi'_{\epsilon'_N}\in {\cal H}^{(s)}_{\epsilon'_N}$ 
where $\epsilon'_N=\kappa^N(\epsilon)$ is the $N$-fold refined theory   
and $P^{(s)}_{\epsilon\epsilon'_N}=
J^{(s)}_{\epsilon\epsilon'_N} [J^{(s)}_{\epsilon\epsilon'_N}]^\dagger$.
Here the fixed $N$ should be as large a number as practically possible.
Then, if these matrix elements become smaller as $s$ increases and as 
$\epsilon$ decreases, one would have strong evidence for convergence to 
the fixed point and vanishing of the anomaly. For $s,N\to \infty$ 
this step becomes step 3.\\
\\
These steps have not been carried out yet for LQG not even in the weakened 
version (\ref{5.4}). However, we want to sketch at least how one starts 
the flow, say for the case that $\sigma$ is compact with periodic 
boundary conditions. By the Bieberbach theorem \cite{34}
we may assume w.l.g.  
that $\sigma$ is a 3-torus.
First we do not consider all possible graphs but only those 
which can be sensibly labelled by a controllable set $\cal E$ and such 
that the discretised classical 
degrees of freedom labelled by $\epsilon \in {\cal E}$ 
still separate the points of the classical phase space. For instance, 
these could be cubical lattices $\gamma_\epsilon$ in $\sigma$ where 
$\epsilon'\le \epsilon$ iff $\gamma_\epsilon$ is a sublattice of 
$\gamma_{\epsilon'}$ and $\epsilon$ could be the lattice spacing with respect 
to some coordinates on $\sigma$. 
We discretise holonomies on the edges of $\gamma_\epsilon$ and fluxes on a 
similarly chosen dual cubical cell complex $\gamma_\epsilon^\ast$. Then we 
take some discretisations $C^I_\epsilon(n_I)$ of the $C^I(n_I)$ where 
\be \label{5.5}
C^I=f^{Ij}\;k\; B^a_j\; e^a_k,\;
B^a_j=\epsilon^{abc} F_{bc}^k\;\delta_{jk},\; 
e_a^j=q_{ab} \delta_{jk}\;\frac{E^b_k}{\sqrt{\det)q)}},\;
f^{Ij}\;_k=\left( \begin{array}{cc}
2\delta^j_k & I=0\\
\epsilon^{ljm}\; \delta_{mk} & I=l
\end{array}
\right.
\ee
are the density weight one 
building blocks of the extended master constraint \cite{10,32}.
For $I=l$ these have been reused more recently with non-standard 
density weight and were called ``electric 
diffeomorphisms'' \cite{21,21a} there. As pointed out in \cite{10,32}
and as follows from the general results established in \cite{15} in contrast 
to $D(u)$ the constraints 
$\vec{C}(\vec{N})=D(u)_{u=\vec{N}\cdot E/\sqrt{\det(q)}}$ which are 
classically equivalent to $D(u)$ {\it for non-degenerate} $q$ can be quantised 
in the LQG representation and by (\ref{5.5}) display a more balanced 
structure as far as the algebraic structure of all 4 constraints is 
concerned and which has the advantage that not only the exponentiated version 
of $D(u)$ exists in the quantum theory.
The price to pay is that now all constraints close with structure functions
only and these are classically well defined only when the 
metric is non-degenerate. We denote them as 
\be \label{5.6}
\{C^I(m_I),C^J(n_J)\}=:C_K(f^K\;_{IJ}(m^I,n^J;q)=:K(m,n)
\ee
We now proceed as in \cite{32}: We introduce holonomy flux variables on 
$\gamma_\epsilon$ and discretise the constraints using them, see \cite{32}
for details, resulting in classical functions $C^I_\epsilon(n_I)$. We
do the same with the right hand side of (\ref{5.6}) resulting in 
$K_\epsilon(m,n)$. By construction we have to leading order in $\epsilon$
\be \label{5.6}
\{C^I_\epsilon(m_I),C^J_\epsilon(n_J)\}=K_\epsilon(m,n)
\ee
and (\ref{5.6}) converges to (\ref{5.5}) pointwise $Z$ on the phase space.

Then we quantise $C^I_\epsilon(m_I), \; K_\epsilon(m,n)$ with 
all flux depending variables ordered to the right \cite{32} and denote 
the resulting operators on ${\cal H}^{(0)}_\epsilon$ by 
$C^{I(0)}_\epsilon(m_I), \; K^{(0)}_\epsilon(m,n)$ that start the 
renormalisation flow. 
Here ${\cal H}^{(0)}_\epsilon$ is $L_2(SU(2),d\mu_H)^{N_\epsilon}$ where 
$\mu_H$ is the Haar measure, $N_\epsilon=|E(\gamma_\epsilon)|$ 
the number of edges of $\gamma$, the fluxes being quantised as right invariant
vector fields on corresponding copies if SU(2) and the holonomies as 
multiplication operators corresponding to that copy. 
Then we wish to study
\be \label{5.7}
\Delta^{(0)}_\epsilon(m,n):=
[C^{I(0)}_\epsilon(m_I),C^{J(0)}_\epsilon(n_J)\}-i\;K^{(0)}_\epsilon(m,n)
\ee
As we showed above, (\ref{5.7}) must not vanish, not even for the 
fixed point family. However, one may hope that (\ref{5.7}) is small
in a suitable operator topology for small $\epsilon$ even for the initial
``naive'' discretisation. We consider the following topoplogy which can 
be argued to be as close as possible to the topology of pointwise convergence 
on the classical on phase space: We consider the coherent states 
$\psi^\epsilon_Z\in {\cal H}^{(0)}_\epsilon$ \cite{22} which take as an 
input a point $Z$ in the classical continuum phase space, map it to elements 
of $g_e\in SL(2,\mathbb{C})$, one for each edge of $\gamma_\epsilon$, take 
a coherent superposition of SNWF over a single edge weighted by 
corresponding irreducible representations of SU(2) anaytically continued to
$SL(2,\mathbb{C})$ and by a Gaussian in the corresponding spin quantuim 
number and finally one takes the tensor product over all edges and 
normalises the result. The construction of the $SL(2,\mathbb{C})$ element 
and the Gaussian factor are not randomly chosen but in fact follow from 
the complexifier machinery \cite{22}. In more detail one constructs  
\be \label{5.7a}
\psi_{e,Z}(A)=\sum_{2j=0}^\infty\; (2j+1)\; e^{-tj(j+1)}\;
{\rm Tr}(\pi_j(g_e(Z)\;H_e(A)^{-1}))
\ee
where $Z=(A^0,E^0)$ is a point in the classical phase space,  
$g_e(Z)=\exp(i E^0_j(S_e)))\;H_e(A_0)\in SL(2,\mathbb{C})$
where $S_e$ is the face in $\gamma_\epsilon^\ast$ dual to $e$. Then
\be \label{5.7b}
\psi_\epsilon(Z):=\prod_e\; \frac{\psi_e}{||\psi_e||}
%\psi_\epsilon(Z):=\prod_e\; \frac{\psi_e-1}{||\psi_e-1||}
\ee
These coherent states are known 
to be sharply peaked on points $Z_\epsilon$ where $Z_\epsilon$ encodes 
the discretised variables constructed from $Z$ and restricted to the edges 
and faces of $\gamma_\epsilon$ and $\gamma_\epsilon^\ast$ respectively
(basically the collection of the $g_e(Z)$).
Also there is a measure $\rho_\epsilon$ on the cotangent bundle phase space 
$\Gamma_\epsilon=
SL(2,\mathbb{C})^{N_\epsilon}
\cong T^\ast(SU(2)^{N_\epsilon})$ coordinatised by $Z_\epsilon$ which leads to 
a resolution of unity 
\be \label{5.8}
\int_{\Gamma_\epsilon}\; d\rho_\epsilon(Z)\; 
\psi_\epsilon(Z)\; <\psi_\epsilon(Z),\;>_{{\cal H}^{(0)}_\epsilon}  
=1_{{\cal H}^{(0)}_\epsilon}  
\ee
Then 
\be \label{5.9}
\Delta^{(0)}_\epsilon(m,n)\psi_\epsilon(Z)     
=\int\; d\rho_\epsilon(Z')\;
<\psi_\epsilon(Z'),\Delta_\epsilon(m,n)\psi_\epsilon(Z)>\;     
\psi_\epsilon(Z')
\ee
%Note that $\Delta_\epsilon(m,n)$ is a Riemann sum where each summand just 
%depends on a few degrees of freedom (those connected to a cubical cell
%and its dual faces), at least in the initial discretisation.
Due to sharp peakedness 
\be \label{5.10}
<\psi_\epsilon(Z'),\Delta^{(0)}_\epsilon(m,n)\psi_\epsilon(Z)>\;     
=<\psi_\epsilon(Z),\Delta_\epsilon(m,n)\psi_\epsilon(Z)>\;     
\;
<\psi_\epsilon(Z'),\Delta^{(0)}_\epsilon(m,n)\psi_\epsilon(Z)>\;     
\ee
plus corrections in $\epsilon$ which are subleading provided 
that the estimates performed for similar operators in \cite{32}
carry over to the present case.  

Then  
\be \label{5.11}
\Delta_\epsilon(m,n)\psi_\epsilon(Z)     
=<\psi_\epsilon(Z),\Delta_\epsilon(m,n)\psi_\epsilon(Z)>\;     
\psi_\epsilon(Z)     
\ee
plus corrections in $\epsilon$ which are subleading provided the above 
assumptions hold. 
It remains to compute the expectation values 
\be \label{5.12}
<\psi_\epsilon(Z),\Delta_\epsilon(m,n)\psi_\epsilon(Z)>\;     
\ee
Again, this kind of calculation has been carried out in \cite{32}
already and one finds 
\be \label{5.13}
<\psi_\epsilon(Z),\Delta^{(0)}_\epsilon(m,n)\psi_\epsilon(Z)>
=[\Delta_\epsilon(m,n)](Z),\;\;
\Delta_\epsilon=\{C^I_\epsilon(m_I),C^J_\epsilon(n_J)\}-K_\epsilon(m,n)
\ee
plus corrections in $\epsilon$ which are subleading provided the above 
assumptions hold. The latter quantity is known to converge to 
zero at fixed $Z$ as $\epsilon$ tends to zero. \\
\\
This shows that, modulo the above reservations, for suffciently
small $\epsilon$ the matrix elements of $\Delta^{(0)}_\epsilon(m,n)$ 
are almost diagonal and approach the classically discretised value, the latter
approaching the classical continuum integral. The validity of the calculation 
and the estimates alluded to {\it rest on the assumtion that $Z$ encodes 
a non-degenerate metric}. Note that coherent states on $\gamma_\epsilon$ 
are non-degenerate 
in the sense of non-vanishing volume expectation values for regions 
containing a vertex of $\gamma_\epsilon$. This is the case 
even if $Z$ is degenerate but 
in this case the $\epsilon$ corrections mentioned above are not subleading.
For details the reader is referred to \cite{22,32}.
This underlines once more the importance
of the non-degeneracy condition even at the quantum level. 

More details of the concrete calculation sketched above will appear 
elsewhere. Note however, that exact closure can only be expected for the 
continuum operator which is the critical theory one tries to produce from the 
renormalisation flow.

\section{Anomaly free, density one parametrised field theory}
\label{s6}

In the first subsection we motivate and define the density 
weight one model, in the second we 
define the LQG inspired Hilbert space representation in which 
quantum non-degeneracy is manifest, in the third 
we define the regularised constraints, in the fourth we remove the
regulator on dual constraints wrt a suitable habitat (space 
of distributions), in the fifth  
we verify the anomaly representation of these dual constraints, 
in the sixth we illustrate how in the degenerate representations 
considered in \cite{25,26} this model would yield an anomalous algebra 
and pin point that indeed the anomaly is {\it caused by degeneracy} and in 
the seventh we comment on the construction of a physical Hilbert space for 
this model.

\subsection{Motivation and definition of the model}
\label{s6.1} 

The constraints of 1+1 dimensional 
parametrised field theory \cite{24,25,26} can be 
written as 
\be \label{6.1}
\tilde{D}=\Pi\;\Phi'+P\; T'+Y\; X',\;\;
\tilde{C}=\frac{1}{2}[\Pi^2+(\Phi')^2]+P\;X'+Y\; T'
\ee
where $T,X$ are the embedding fields with conjugate momenta $P,T$, the 
massless Klein-Gordon field pulled back by $T,X$ and its conjugate momentum
are denoted as $\Phi,\Pi$. A prime denotes derivation with respect to 
the angular variable $x\in [0,1)$ and a dot derivation with respect to the 
time foliation parameter $t\in \mathbb{R}$. We have e.g.
\be \label{6.2}
\{\Pi(x),\Phi(y)\}=\delta(x-y)
\ee
where $\delta(x)$ is the 1-periodic delta-distribution. The constraints 
satisfy the classical hypersurface deformation algebra $\mathfrak{h}$
\be \label{6.3}
\{\tilde{D}(u),\tilde{D}(v)\}=-\tilde{D}([u,v]),\;\; 
\{\tilde{D}(u),\tilde{C}(\tilde{n})\}=-\tilde{C}([u,\tilde{n}]),\;\; 
\{\tilde{C}(m),\tilde{C}(n)\}=-\tilde{D}([\tilde{m},\tilde{n}]),\;\;
[u,v]=u\;v'-v\; u'
\ee
This does not resemble the form that $\mathfrak{h}$ has in GR. 
There are two reasons for this which are due to the pecularity of two 
spacetime dimensions. First, the constraints (\ref{6.1}) have density weight 
two rather than one, because tensors of rank $(a,b)$ are scalar 
densities of weight $b-a$. Second, (\ref{6.3}) does not display any 
structure functions. This is because a spatial metric $q$ is just a scalar 
density of weight two, hence the density weight one scalar constraints 
would be $C=\sqrt{q}^{-1} \tilde{C}$ and their Poisson brackets would yield  
$q^{-1} D=q^{-2} \tilde{D}$ instead, 
which explains why there is no $q$ dependence
in the Poisson brackets of the $\tilde{C}$. Note also that the smearing 
functions are naturally vector fields i.e. scalar densities of weight $-1$.

Accordingly, to shed light on the complex of questions that concerns us 
in the present work - density weights, structure functions, anomaly freeness,
non-degeneracy, habitats - the form of the constraints (\ref{6.1})   
is not useful. To make the analogy with (Euclidian) GR manifest we 
relabel the canonical pairs
\be \label{6.4}    
A_1:=T,\; E^1:=P;\; A_2:=X,\; E^2:=Y
\ee
and rewrite the constraints in these variables however with density 
weight unity for the Hamiltonian constraint
\be \label{6.5}  
D:=\Pi\;\Phi'+A_1'\; E^1+A_2' \; E^2,\;
C=[\frac{1}{2}[\Pi^2+(\Phi')^2]+E^1\;A_2'+E^2\; A_1']\;|E^1 E^2|^{-1/2}
\ee
There is no curvature in one dimension but $A_1',A_2'$ can be 
considered as a 
substitute depending like a curvature 
on the derivative of the ``connections'' $A_1,A_2$
(which are actually scalar fields). We also have introduced the 
density two valued metric $q:=|E^1 E^2|$. The constraints (6.1) and 
(\ref{6.2}) are classically equivalent iff $q>0$ i.e. if that metric is 
{\it non-degenerate}.

We compute (note that $u,v$ are scalar densities of weight $-1$ while
$m,n$ are scalar densities of weight zero) 
\ba \label{6.6}
&& \{D(u),D(v)\} = -D([u,v]),\;\; 
\{D(u),C(n)\}=-C(u[n]),\;\; 
\nonumber\\
&& \{C(m),C(n)\}=-{D}([m\; dn-n\;dm] \; q^{-1})+
C([m\; dn-n\;dm]\;\frac{1}{2}(
\frac{E^1}{E^2}+\frac{E^2}{E^1})q^{-1/2})
\ea
which resembles $\mathfrak{h}$ of GR more closely than (\ref{6.3}) 
because $u,v,D$ and $m,n,C$ assume their standard density weight and because
the $C(n)$ close with non-trivial structure functions. In some sense these 
structure functions are even more complicated than the ones of GR, hence 
the quantisation of this model in a LQG inspired representation will be 
a rather stringent test of the validity of the viewpoint that 
density weight one and non degeneracy are intimately connected while 
yielding anomaly freeness on suitable habitats.

\subsection{LQG inspired Hilbert space representation}
\label{s6.2}

In its density weight two versions, the constraints (\ref{6.1})  
are naturally quantised in a Fock representation. This option 
is not avaliable for the density weight one version (\ref{6.5}) because 
an operator valued distribution such as 
\be \label{6.7}    
Q:=|E^1 E^2|^{1/2}
\ee
is ill-defined in that Fock represention, at least as far as the geometry 
sector is concerned. We thus adopt the ``hybrid quantisation approach''
employed in Loop Quantum Cosmology (LQC) \cite{35} and consider the usual
Fock representation ${\cal H}_F$ for the matter sector 
\cite{24,33} and an LQG inspired representation ${\cal H}_G$
for the geometry sector. The total representation space is then
the tensor product ${\cal H}={\cal H}_F\otimes {\cal H}_G$.

As motivated in section \ref{s1} the corresponding geometry vacuum 
$\Omega_G$ is 
annihilated by the ``electric flux operators'' $E[f]$ and the representation 
is discontinuos with respect to the ``holonomy operators'' $H[g]$ where 
$f,g$ are a pair of scalar smearing function and ($I=1,2$)
\be \label{6.8}    
E[f]=\int_{[0,1)}\; dx\; f_I(x)\; E^I(x),\;\;     
H[g]=\exp(i\int_{[0,1)}\; dx\; g^I(x)\; A_I'(x),\;\;     
\ee
That the derivative of $A_I'$ instead of $A_I$ appears in $H[g]$
is justified 
by the fact that the spatial manifold $\sigma=[0,1)$ is a loop, i.e. 
a circle, hence by ``Stokes theorem'' the holonomy
$H(g)$ is the exponential of the ``magnetic flux''. Explicitly we have 
in terms of the original variables that
\be \label{6.8a} 
<g^I,A'_I>_{L_2([0,1),dx)}=-<g^{I\prime},A_I>
+T(1)\;g^1(1)-T(0)\;g^1(0)
+X(1)\;g^2(1)-X(0)\;g^2(0)
\ee
Physically we have $T(1)=T(0),\; X(1)=X(0)+1$ as $X$ is an angular coordinate.
If we impose, as we will do in what follows, 
that $T(0),X(0)$ assume fixed values then the boundary term of the variation 
of (\ref{6.8a}) vanishes even if $g^I$ is not periodic as we will consider 
below. Thus we may use that $\delta<g^I,A_I'>=-<g^{I\prime},\delta A_I>$.
 
Then the representation is compeletly defined by the relations
\ba \label{6.9}
&& E[f]^\ast=E[f],\;H[g]^\ast=H[-g],\; H[g]\; H[\tilde{g}]=H[g+\tilde{g}],\;
[E[f],E[\tilde{f}]]=0,\;
[H[g],H[\tilde{g}]=0,\;
\nonumber\\
&& [E[f], H[g]]=<f,g'>\; H[g],\;
E[f]\Omega_G=0,\;
<\Omega_G, H[g]\Omega_g>=\delta_{g,0}
\ea
which clearly resembles the LQG representation. The vector $\Omega_G$
is cyclic and the span of the ``Weyl states''
$H(g)\Omega_G$ is dense (thus replacing the spin network functions).

In this representation 
it is much simpler to find non-degenerate states according
to the definition in section \ref{s4}, than in the LQG representation. 
This is because the Weyl states
$H[g]\Omega_G$ already diagonalise the ``volume operator'' 
\be \label{6.10}
V(O):=\int_O\; dx \; Q(x)
\ee
for any open $O\subset [0,1)$ namely 
\be \label{6.11}
V(O)\;H[g]\Omega_G=V_g(O),\;H[g]\Omega_G,\;\;
V_g(O):=\int_O\; dx \; Q_g(x),\;
Q_g(x):=|g^{1\prime}\;g^{2\prime}|^{1/2}(x)
\ee
thus $H(g)\Omega_G$ is non-degenerate iff $g^{1\prime}, g^{2\prime}$ 
are nowhere vanishing in $[0,1)$. We call such $g$ also non-degenerate.
Thus non-degenerate $g$ are strictly monotonous and thus are not periodic 
but ``angular functions'' which is why we required $\delta A_I(0)=0$ above.

\subsection{Constraint regularisation}
\label{s6.3}

We formally apply (\ref{6.5}), with all dependence on $E^1,E^2$ 
ordered to the right, to tensor products of Fock states with 
Weyl states with non-degenerate $g$
\ba \label{6.12}
&& D[u]\;\psi_F\otimes \;H[g]\Omega_G
=[\int\; dx\; u(x)\; 
[d(x)\otimes 1_G+1_F\otimes \{
A_1'\; g^{1\prime} 
+A_2'\; g^{2\prime}\}(x)]\; \psi_F\otimes \;H[g]\Omega_G
\nonumber\\
&& C[n]\;\psi_F\otimes \;H[g]\Omega_G
=[\int\; dx\; n(x)\; 
[h(x)\otimes 1_G+1_F\otimes \{
A_1'\; g^{2\prime}
+A_2'\; g^{1\prime}\}(x)]\;Q_g(x)^{-1}]\;
\psi_F\otimes \;H[g]\Omega_G 
\nonumber\\
&& d(x)=:\;\Pi(x)\Phi'(x)\;:,\;
h(x)=:\;\frac{1}{2}[\Pi(x)^2+(\Phi'(x))^2]\;:
\ea
where $:\;\;:$ denotes the Fock space normal ordering. 

As the representation is irregular with respect to the connection, the 
objects $A_I'(x)$ in (\ref{6.12})
do not even exist as operator valued distributions. We 
therefore use the classical identity
\be \label{6.13}
A_I'(x)= \lim_{N\to \infty}\;\lim_{s\to 0} 
\frac{H[j_{Ix}^{s,N}]-1}{is},\;
[j_{Ix}^{s,N}]^J(y)=s\; \delta_I^J\delta_N(x,y),\;
\delta_N(,y)=\sum_{k\in \mathbb{Z},|k|\le N}\;e^{2\pi\;i\;(x-y)}
\ee
to regularise (\ref{6.12}) as 
\ba \label{6.14}
&& D_{s,N}[u]\;\psi_F\otimes \;H[g]\Omega_G
=[\int\; dx\; u(x)\; 
[d(x)\otimes 1_G+1_F\otimes \{
g^{1\prime} \frac{H[j_{1\cdot}^{s,N}]-1}{is} 
+g^{2\prime}\; \frac{H[j_{2\cdot}^{s,N}]-1}{is} \}(x)]\; 
\psi_F\otimes \;H[g]\Omega_G
\nonumber\\
&& \hat{C}_s[n]\;\psi_F\otimes \;H[g]\Omega_G
=[\int\; dx\; n(x)\; 
[h(x)\otimes 1_G+1_F\otimes \{
g^{2\prime}\; \frac{H[j_{1\cdot}^{s,N}]-1}{is} 
+g^{1\prime}\;\frac{H[j_{2\cdot}^{s,N}]-1}{is} \}(x)]\;Q_g(x)^{-1}]
\times
\nonumber\\
&& \psi_F\otimes \;H[g]\Omega_G 
\ea
This expression is still formal because the Lebesgue integral over vectors 
of the form $F(x)\;H[g_x]\Omega_G$ with $F$ a continuous function 
and $g_x=g_y$ iff $x=y$ has zero norm in ${\cal H}_G$. We thus in addition
introduce a Riemann sum approximation of the inegral by intervals $\Box$
of coordinate size $\epsilon$ and centre $p_\Box$ and obtain
\ba \label{6.15}
&& D_{s,N,\epsilon}[u]\;\psi_F\otimes \;H[g]\Omega_G
=\epsilon\sum_\Box \;[ u\; 
[d\otimes 1_G+1_F\otimes \{
g^{1\prime} \frac{H[j_{1\cdot}^{s,N}]-1}{is} 
+g^{2\prime}\; \frac{H[j_{2\cdot}^{s,N}]-1}{is} \}]]_{x=p_\Box}\; 
\psi_F\otimes \;H[g]\Omega_G
\nonumber\\
&& \hat{C}_{s,\epsilon}[n]\;\psi_F\otimes \;H[g]\Omega_G
=\epsilon\sum_\Box \;[ n\; 
[h\otimes 1_G+1_F\otimes \{
g^{2\prime}\; \frac{H[j_{1\cdot}^{s,N}]-1}{is} 
+g^{1\prime}\;\frac{H[j_{2\cdot}^{s,N}]-1}{is} \}]\;Q_g^{-1}]_{x=p_\Box}\;
\psi_F\otimes \;H[g]\Omega_G 
\nonumber\\
&&
\ea
and the sum is finite as $[0,1)$ is compact. 

\subsection{Regulator removal}
\label{s6.4}

We consider $\cal D$, the finite linear span of the $b\otimes w[g]$
where the countable system of states $b$ provides an orthonormal Fock 
basis of ${\cal H}_F$ and $w[g]:=H[g]\Omega_G$. Then an algebraic 
distribution $l\in {\cal D}^\ast$ is of the form
\be \label{6.16}
l=\sum_{b,g}\; l(b,g)\; <b\otimes w[g],\;.>_{{\cal H}}
\ee
Its action on (\ref{6.15}) is 
\ba \label{6.17}
&&l[D_{s,N,\epsilon}[u]\; b\otimes w[g]]
=\epsilon\sum_{\Box,\hat{b},\hat{g}} \;l(\hat{b},\hat{g})\;[ u\; 
[<\hat{b},d\; b>\; \delta_{\hat{g},g}+\delta_{\hat{b},b}\; \{
g^{1\prime} \frac{\delta_{\hat{g},g+j_{1\cdot}^s}-\delta_{\hat{g},g}}{is} 
+g^{2\prime} \frac{\delta_{\hat{g},g+j_{2\cdot}^s}-\delta_{\hat{g},g}}{is} 
\}]]_{x=p_\Box}\; 
\nonumber\\
&=& \epsilon\sum_{\Box} \;[ u\;[\sum_{\hat{b}}\; <\hat{b},\; d\;b>
l(\hat{b},g)
+g^{1\prime}\;\frac{l(b,g+j_{1\cdot}^{s,N})-l(b,g)}{is} 
+g^{2\prime}\;\frac{l(b,g+j_{2\cdot}^{s,N})-l(b,g)}{is}]]_{x=p_\Box} 
\nonumber\\
&&l[C_{s,N,\epsilon}[n]\; b\otimes w[g]]
= \epsilon\sum_{\Box} \;[ n\;\{[\sum_{\hat{b}}\; <\hat{b},\; h\;b>
l(\hat{b},g)]
\nonumber\\
&& +g^{1\prime}\;\frac{l(b,g+j_{2\cdot}^{s,N})-l(b,g)}{is} 
+g^{2\prime}\;\frac{l(b,g+j_{1\cdot}^{s,N})-l(b,g)}{is}\}\;
Q_g^{-1}]_{x=p_\Box} 
\ea
Taking $\epsilon\to 0$ first returns the Riemann sum into an integral.
Then taking $s\to 0$ gives 
\be \label{6.18}
\lim_{s\to 0}\; \frac{l(b,g+j_{I\cdot}^{s,N})-l(b,g)}{is}  
=-i\int\; dy\; \delta_N(x,y)\; \frac{\delta l(b,g)}{\delta g^I(y)}
\ee
i.e. the integral of the regularised and smooth $\delta$ distribution
against the functional derivative of $l$. Then taking $N\to \infty$
reduces (\ref{6.18}) to the functional derivative at $x$. Accordingly,
taking the three limits in that order defines the dual constraints
\ba \label{6.19}
&& 
(D'[u]\; l)(b\otimes g)=\sum_{\hat{b}}\; <\hat{b},\; d[u]\;b>\; l(\hat{b},g)
-i\int\;dx\;u(x)\;[
g^{1\prime}(x)\;\frac{\delta l(b,g)}{\delta g^1(x)}
+g^{2\prime}(x)\;\frac{\delta l(b,g)}{\delta g^2(x)}]
\\
&& (C'[n]\; l)(b\otimes g)=\sum_{\hat{b}}\; <\hat{b},\; h[n]\;b>\; l(\hat{b},g)
-i\int\;dx\;n(x)\;[
g^{1\prime}(x)\;\frac{\delta l(b,g)}{\delta g^2(x)}
+g^{2\prime}(x)\;\frac{\delta l(b,g)}{\delta g^1(x)}] Q_g(x)^{-1}
\nonumber
\ea
Introducing the notation 
\be \label{6.20}
<u,g^{I\prime}\;\delta_{g^J}>:=
\int \; dx\; u(x)\; g^{I\prime}(x)\; \frac{\delta}{\delta g^J(x)}
\ee
we may write the dual operations in the compact form
\ba \label{6.21}
D'[u] \; l=
\sum_{b,g}\;
[\sum_{\hat{b}}\; <d[u]\hat{b},\;b>\; l(\hat{b},g)
-i \;<u,\;(g^{1\prime}\;\delta_{g^1}+g^{2\prime}\;\delta_{g^2}>\;
l(b,g)]\;<b\otimes w[g],.>
\nonumber\\
C'[n] \; l=\sum_{b,g}\;
[\sum_{\hat{b}}\; <h[n]\hat{b},\;b>\; l(\hat{b},g)
-i \;<n\;Q_g^{-1},\;(g^{1\prime}\;\delta_{g^2}+g^{2\prime}\;\delta_{g^1}>\;
l(b,g)]
\;<b\otimes w[g],.>
\ea

\subsection{Algebra of dual constraints}
\label{s6.5}

The domain of definition of $D'[u],\; C'(n)$ is given by functionals 
$l(b,g)$ which are functionally differentiable with respect to $g$
and which either have support on non-degenerate $g$ or, as we 
motivated frequently in this paper, are only tested 
with respect to non-degenerate $g$. In order that we can 
compute commutators we need this domain to be invariant, hence 
$l$ should at least be twice functionally differentiable and its support 
should not be changed by taking faunctional derivatives or 
alternatively also higher derivatives 
should be tested with non-degenerate $g$ only.
This domain is certainly non-trivial, any smooth function $F$ in $N$
variables of the form $<j_k,g>$ for smooth functions $j_k,\;k=1,..,N$ 
is in this domain, even under any order of constraint actions.\\
\\
{\bf Diffeomorphism-Diffeomorphism}\\
\\
We have 
\ba \label{6.21}
&& [D'[u]\;D'[v]\;l](b,g)
= \sum_{\hat{b}}\; <d[u]\hat{b},\;b>\; [D'[v]\;l](\hat{b},g)
-i \;<u,\;(g^{1\prime}\;\delta_{g^1}+g^{2\prime}\;\delta_{g^2}>\;
[D'[v]\;l](b,g)
\nonumber\\
&=&
\sum_{\hat{b}}\; <d[u]\hat{b},\;b>\;
\{
\sum_{\tilde{b}}\; <d[v]\tilde{b},\;\hat{b}>\; l(\tilde{b},g)
-i \;<v,\;(g^{1\prime}\;\delta_{g^1}+g^{2\prime}\;\delta_{g^2}>\;
l(\hat{b},g)
\}
\nonumber\\
&& 
-i \;<u,\;(g^{1\prime}\;\delta_{g^1}+g^{2\prime}\;\delta_{g^2}>\;
\{
\sum_{\hat{b}}\; <d[v]\hat{b},\;b>\; l(\hat{b},g)
-i \;<v,\;(g^{1\prime}\;\delta_{g^1}+g^{2\prime}\;\delta_{g^2}>\;
l(b,g)
\}
\nonumber\\
&=&
\sum_{\hat{b}}\; <d[v]\hat{b},\; d[u]\;b>\; l(\hat{b},g)
\nonumber\\
&& -i \;\sum_{\hat{b}}\; \{
<d[u]\hat{b},\;b>\;
<v,\;(g^{1\prime}\;\delta_{g^1}+g^{2\prime}\;\delta_{g^2}>\;
+
<d[v]\hat{b},\;b>\;
<u,\;(g^{1\prime}\;\delta_{g^1}+g^{2\prime}\;\delta_{g^2}>\}
l(\hat{b},g)
\nonumber\\
&& -\;
<u,\;(g^{1\prime}\;\delta_{g^1}+g^{2\prime}\;\delta_{g^2}>\;
<v,\;(g^{1\prime}\;\delta_{g^1}+g^{2\prime}\;\delta_{g^2}>\;
\;l(b,g)
\ea
where we used the completeness relation on the Fock basis $b$ and the 
symmetry of the normal ordered Fock operators $d[u],\;d[v]$. Upon taking 
the commutator we see that the second and third term in (\ref{6.21}) cancel
as they are symmetric in $u,v$. The first terms in (\ref{6.21}) combine 
into 
\be \label{6.22}
\sum_{\hat{b}}\; <\hat{b},\;[d[v],d[u]]\;\;b>\; l(\hat{b},g)
\ee
where we made again use symmetry of the Fock operators. The fourth terms
can be worked out assuming that second functional derivatives commute,
using the fundamental functional derivatives 
\be \label{6.23}
\frac{\delta g^{I\prime}(x)}{\delta g^J(y)}=\delta^I_J\; 
\partial_x\delta(x,y)
\ee
and that integrations by parts does not create boundary terms which can be 
granted e.g. by assuming that $u,v$ vanish there. Then the commutator 
of the fourth terms is found to be
\be \label{6.24}
<u\;v'-v'\; u,(g^{1\prime}\;\delta_{g^1}+g^{2\prime}\;\delta_{g^2}>\;
l(b,g)
\ee  
Since in the Fock representation we have by construction
\be \label{6.25}
[d[v],d[u]]=-i\;(d[v\; u'-u\;v']+c_{DD}(v,u))
\ee
with the central term $c_{DD}(v,u)$ of the Virasoro algebra
one finds altogether
\be \label{6.26}
[D'[u],D'[v]]=i\;(D'[u\;v'-v\;u']+c_{DD}(u,v))
\ee
which is an {\it anti-}representation of the diffeomorphism algebra. 
This is because taking commutators in the dual space reverses order. Note 
also that we only obtain one central term, not three. This is because 
the geometrical sector is not normal ordered with respect to the 
Fock anihilators but rather wrt the geometric ``annihilators'' $E^I$. \\
\\
{\bf Diffeomorphism-Hamiltonian}\\
\\
The other commutators are more complicated to compute but follow the same 
pattern. We have using again completeness and symmetry
\ba \label{6.27}
&& [D'[u],C'[n]](b,g)
= \sum_{\hat{b}}\; <\hat{b},\;[h[\frac{n}{Q_g}],d[u]]\;\;b>\; l(\hat{b},g)
\nonumber\\
&& -i\sum_{\hat{b}}\; <\hat{b},\;
h[<u,g^{1\prime}\;\delta_{g^1}+g^{2\prime}\;\delta_{g^2}>,\frac{n}{Q_g}]\;b>
\; l(\hat{b},g)
\nonumber\\
&& -[<u,\;g^{1\prime}\;\delta_{g^1}+g^{2\prime}\;\delta_{g^2}>\;
,<\frac{n}{Q_g},g^{1\prime}\;\delta_{g^2}+g^{2\prime}\;\delta_{g^1}>]\;l(b,g)
\ea
As compared to the previous calculation one needs 
\be \label{6.28}
\frac{\delta}{\delta g^I(x)}\;|g^{J\prime}(y)|^{-1/2}
=-\frac{\delta_I^J}{2}\;\frac{1}{g^{J\prime}(y)\;|g^{J\prime}(y)|^{1/2}}
\; [\partial_y\delta(x,y)]
\ee
Using the Virasoro algebra 
\be \label{6.29}
[d[u],h(\tilde{n})]=-i(h(u\; \tilde{n}'-u'\; \tilde{n})+c_{DC}(u,\tilde{n})
\ee
and (\ref{6.28}) one finds that the first and second term after 
several cancellations of terms combine to 
\be \label{6.29}
i[\sum_{\hat{b}}\; <\hat{b},\;h[u\; \frac{n'}{Q_F}]\;b> \;l(\hat{b},g)
+c_{DC}(u,\frac{n}{Q_g})\; l(b,g)]
\ee
Note that (the prime denotes dual action and not derivation w.r.t. $x$0
\be \label{6.30}
l[b\otimes V(O)^{-1} w(g)]=V_g(O)^{-1}\; l(b,g) 
=[[1_F\otimes V(O)^{-1}]'\;l][b\otimes w(g)]
\ee
The third term in (\ref{6.27}) yields using (\ref{6.28}) after a longer 
but elementary calculation
\be \label{6.31}
<u\frac{n'}{Q_g},\;
g^{1\prime}\;\delta_{g^2}+g^{2\prime}\;\delta_{g^1}>\;l(b,g)
\ee
Thus altogether
\be \label{6.32}
[D'[u],C'[n]]=i\;[C'(u\;n')+c_{DC}(u,n\; [Q']^{-1})]
\ee
Of independent interest is the fact that the central term is no longer 
central but depends on the inverse dual volume density $Q'$ as follows
from (\ref{6.28}) and (\ref{6.29}). \\
\\
{\bf Hamiltonian-Hamiltonian}\\
\\
We have by already familiar methods
\ba \label{6.33}
&& [[C'(m),C'(n)]\;l](b,g)
=\sum_{\hat{b}}\; <\hat{b},\;[h(\frac{n}{Q_g}),\;h(\frac{m}{Q_g})]\;b>
\; l(\hat{b},g)
\nonumber\\
&& +i\;\sum_{\hat{b}}\;\{
<h([<\frac{n}{Q_g},g^{1\prime}\delta_{g^2}+g^{2\prime}\delta_{g^1}>,\;
\frac{m}{Q_g}])\;\hat{b},\;b>
-
<h([<\frac{m}{Q_g},g^{1\prime}\delta_{g^2}+g^{2\prime}\delta_{g^1}>,\;
\frac{n}{Q_g}])\;\hat{b},\;b>
\}\;l(\hat{b},g)
\nonumber\\
&&
- 
[<\frac{m}{Q_g},g^{1\prime}\delta_{g^2}+g^{2\prime}\delta_{g^1}>  
<\frac{n}{Q_g},g^{1\prime}\delta_{g^2}+g^{2\prime}\delta_{g^1}>]\;
l(b,g)
\ea
Using the Virasoro algebra
\be \label{6.34}
[h(\tilde{m}),h(\tilde{n})]=
i\;[d(\tilde{m}'\tilde{n}-\tilde{n}'\tilde{m})+c_{CC}(\tilde{m},\tilde{n})]
\ee
the first term in (\ref{6.33}) becomes
\be \label{6.35}  
i\sum_{\hat{b}}\; 
[<\hat{b},\;d(\frac{m\;n'-m'\;n}{Q_g^2})
-c_{CC}(\frac{m}{Q_g}),\;h(\frac{n}{Q_g})]\;b>
\; l(\hat{b},g)
\ee
The middle terms in (\ref{6.33}) are computed using (\ref{6.28}) to be 
\be \label{6.36}
-\frac{i}{2}\sum_{\hat{b}}\;
<\hat{b},\;
h[(m\; n'-m'\;n)\;
\frac{g^{1\prime}}{g^{2\prime}}+\frac{g^{2\prime}}{g^{1\prime}}]\;Q_g^{-1}]
\;b>\; l(\hat{b},g)
\ee
Finally, a tedious but straightforward calculation which uses (\ref{6.28}) 
yields the last term in (\ref{6.33}) to be 
\be \label{6.37}
<\frac{m\;n'-m'\;n}{Q_g^2}\;
[1-\frac{1}{2}
\frac{g^{1\prime}}{g^{2\prime}}+\frac{g^{2\prime}}{g^{1\prime}}],\;
[g^{1\prime}\; \delta_{g^2}+g^{2\prime} \; \delta_{g^1}]>\;\;
l(b,g)
\ee
Combining all three contributions we get
\be \label{6.38}
[C'(m),C'(n)]=
i\; [
D'(\frac{(m\;n'-m'\;n)}{[Q']^2})
+c_{cc}(\frac{m}{Q'},\frac{n}{Q'})
-\frac{1}{2}\;C'(\frac{(m\;n'-m'\;n)}{Q'}\;
[\frac{E^{1\prime}}{E^{2\prime}}+\frac{E^{2\prime}}{E^{1\prime}}])
]
\ee 
where $Q',E^{I\prime}$ denotes dual action of $Q,E^I$ and it is understood 
that these objects are ordered to the {\it outmost left} when acting on 
a distribution $l$ in order to correctly reproduce (\ref{6.36}). 

\subsection{Anomaly freeness and non-degeneracy}
\label{s6.6}
 
Comparing the classical relations (\ref{6.6}) with (\ref{6.26}), 
(\ref{6.32}) and (\ref{6.38}) we see that we get {\it an anomaly 
free anti-representation on non-degenerate distributions} including 
the correct central extensions caused by the Fock normal ordering 
of the matter sector. If one
goes through the quantum calculation in detail, it is reassuring to see 
that all of the classical 
Poisson bracket relations that were used in the derivation 
of (\ref{6.6}) are being reused in the quantum computation. \\
\\
As the computation so far is only meaningful on non-degenerate distributions 
one may wonder whether it can be extended to degenerate distributions. 
A drastic case is the LQG like representation of PFT considered in 
\cite{25,26} which is a {\it purely degenerate representation}. In the
language used here, the functions $g^I$ used in \cite{25,26} were not 
smooth but rather piecewise constant, namely characteristic functions of 
intervals of the circle multiplied by constants, that is, step functions
with support on left closed and right open intervals partitioning 
$[0,1)$.
The derivative of a step function is zero almost everywhere and has 
a delta distribution singularity at the interval ends. Accordingly the 
corresponding holonmomies actually are excited only at finitely many points 
$v\in V(g)$ where $V(g)$ denotes the ``vertices'' of the step 
function $g$. If $O$ is an open interval 
of the circle then on such a Weyl function $w[g]$ we get the eigenvalue 
\be \label{6.39}
V_g(O)=\sum_{v\in O\cap V(g)}\; V_g(v),\;
V_g(v):=|\bar{s}^1_v\;\bar{s}^2_v|^{1/2},\;\;
\bar{s}^I_v=g^I(v')-g^I(v)
%\lim_{\delta\to 0+}\;\frac{1}{2}[
%g^I(v-\delta)+g^I(v+\delta)] 
\ee
where $v'\in V(g)$ is the left neighbour of $v\in V(g)$.
In order to define the operator $V(O)^{-1}$ we may e.g.
use Poisson bracket identities 
as in \cite{26} or, for the purposes of this paper sufficient, 
simply Tychonov regularisation 
$V(O)^{-1}:=\lim_{\delta\to 0}\frac{V(O)}{V(O)^2+\delta^2}$. As a result
the eigenvalues of $V(O)$ on $w[g]$ are simply given by $V_g(O)^{-1}$ if 
$V_g(O)>0$ and zero otherwise. Let us use this definition of $V(O)^{-1}$
in (\ref{6.12}) and focus on the geometry part 
of both $D(u),C(g)$. Then the operators $E^I$ 
applied to $w[g]$ restrict the integrals involved
in $D(u),C(n)$ to a sum over $v\in V(g)$. This requires to regularise
the objects 
\be \label{6.40}
\sum_v \;u(v)\;\;\sum_I A_I'(v) s^I_v\; w[g],\;\;
\sum_v\; n(v)\; \frac{A_1'(v)\; s^2_v+A_2'(v)\;s^1_v}{Q(v)} w[g]
\ee
If $O_\epsilon(v)$ is an open interval of coordinate length $\epsilon$ and 
centre $v$, a possible regularisation of (\ref{6.40}) in terms of step 
function smearing functions is 
\be \label{6.41}
\frac{1}{i\epsilon}\;\sum_v \;u(v)\;\;
[w[g_{\epsilon,v})]-1]\; w[g],\;\;
-i\sum_v\; \frac{n(v)}{Q_g(v)}\; 
[w[\tilde{g}_{\epsilon,v})]-1]\; w[g],\;\;
\ee
where 
\be \label{6.42}
g^I_{\epsilon,v}(x)=\chi_{O_\epsilon(v)}(x) s^I_v,\;
\tilde{g}^1_{\epsilon,v}(x)=\chi_{O_\epsilon(v)}(x) s^2_v,\;
\tilde{g}^2_{\epsilon,v}(x)=\chi_{O_\epsilon(v)}(x) s^1_v
\ee
and $\chi_O$ is the characteristic function of $O$. Thus (\ref{6.41}), which 
replaces (\ref{6.15}) maps the span $\cal D$ of $w[g]$ with step functions $g$ 
to itself. 

The algebraic dual of this $\cal D$ consists of distributions 
\be \label{6.43}
l=\sum_g\; l(g)\; <w[g],.>
\ee
where the sum is now restricted to step functions. Computing the 
duals of (\ref{6.41}) on (\ref{6.43}) and taking the limit $\epsilon\to 0$
yields a finite result for $D'(u)$ 
if $l$ is differentiable in the sense that
$\lim_{\epsilon\to 0}\frac{l(g+g_\epsilon(v))-l(g)}{\epsilon}$ exists.
But then $C'(n)\equiv 0$. This happens, precisely because $w[g]$ is 
degenerate almost everywhere. Indeed, we see that in the Riemann sum 
approximation (\ref{6.15}) for sufficiently small $\epsilon$ 
only those cells $\Box$ contribute which contain precisely one vertex 
$v\in V(g)$ thus resulting in (\ref{6.41}). The missing sum over all
cells when taking $\epsilon\to 0$ on the dual, 
which returned the Riemann sum  
into an integral in passing from (\ref{6.17}) to (\ref{6.18}), is missing 
here and prevents a non-trivial action of $C'(n)$. This is a 
{\it concrete demonstration of the importance of non-degeneracy}.

One may argue that in this model one should therefore use the density 
two constraints $\tilde{D},\tilde{C}$. Indeed, as is well known, by taking   
linear combinations, those constraints are equivlent to two commuting 
diffeomorphism constraint algebras. While this is true, the purpose of 
this model with standard density weight for the constraints was to 
mimick the situation in GR as close as possible where for no choice 
of density weight it is possible to avoid the structure functions and where 
density weight one is uniquely selected as the universal choice as 
demonstrated in section \ref{s3}. 
 
We could now proceed as in section \ref{s4} and use linear combinations 
of parallel translates 
of step function states $w[g]$ to a dense set of points on the circle
with coefficients such that the linear combination is normalisable. 
This would then be non-degenerate states which do not leave the realm of 
step functions as smearing functions in the Weyl operators. We leave this 
for future investigations but remark that working with smooth and non
degenerate $g$ rather than step functions is at least much more convenient.

\subsection{Solutions and physical Hilbert space}
\label{s6.7} 

Given the explicit form of the dual constraints (\ref{6.21}) on the 
domain of $l$ with $l(b,g)$ functionally differentiable wrt $g$ we ask 
for solutions $l$ and a Hilbert space structure thereon. In contrast 
to a Fock quantisation of $\tilde{D},\tilde{C}$ the above domain does not 
carry an obvious Hilbert space structure. While not necessary, see
below, we may supply such a structure 
as follows: The set of labels $b$ is discrete as Fock spaces are separable
while the set of labels $g$ is continuous. We may therefore consider 
a Hilbert space of coefficients $l(b,g)$ which are square summable wrt 
$b$ and square integrable wrt $g$ wrt a measure $\mu$ on the space 
${\cal G}$ of 
$g$ which is supported on non-degenerate $g$. That is, the space 
of $l$ 
may be given the Hilbert space structure 
${\cal H}_F\otimes L_2(d\mu,{\cal G})$ where ${\cal H}_F$ is the matter 
Fock space. 

With respect to such a 
measure the $D'(u), C'(n)$ are not symmetric operators and they should not 
be because of the non-trivial structure functions by the argument given 
in \cite{36}. Due to 
the central extension, there can be no strong solutions with respect to 
such a Hilbert space structure which is also true for the usual density 
weight two Fock space quantisation. However, one can construct weak solutions
in the usual way because the matter parts $d,h$ of the constraints are 
still symmetric and it is only those that cause the central extension.     
As the purpose of the present model was just to illustrate the {\it drastic
effect of non-degeneracy} we do not go into further details here and leave 
the issue of solutions for possible future research.

In closing, we remark that the procedure of first having defined 
a representation of the CCR among $A_I, E^I$ which is irregular for 
$A_I$ and a vacuum representation for $E^I$ just to find out that a 
dual represention in which $E^I$ acts by multiplication by $g^{I\prime}$ 
and $A_I'$ by functional derivation with respect to $g^I$ is better suited 
to formulate the dynamics apparently could have been avoided altogether as one 
could have started with that latter representation right away. The 
point is, however, if we had started with the ``g-representation'' then we 
should have had to worry about the measure $\mu$ on $\cal G$ right away and 
we should have had to make $E^I, A_I$ self-adjoint with respect to that 
measure. The indirect method frees us from doing that because a Hilbert 
space structure is only needed on the space of solutions of the constraints.

\section{Conclusion and Outlook}
\label{s7}

In the present work we have illustrated that in canonical quantum gravity 
density weight one of the constraints is not only natural as far as 
spatial diffeomorphism covariance is concerned but also dynamically selected.

We have shown that it comes at no surprise that testing the density weight one 
algebra on degenerate states leads to anomalous (dual,
i.e. on spaces of distributions) representations 
thereof. 

Furthermore, we have shown in an example that an anomaly free implementation of 
the hypersurface deformation algebra $\mathfrak{h}$ in natural representations 
of the CCR in which the spatial metric annihilates the vacuum (which is 
therefore the most degenerate state imaginable) is still possible 
if one does not neglect the fact that in the classical theory the algebra 
$\mathfrak{h}$ only {\it exists} when the metric is non-degenerate. 
Thus, the new point of view advertised here is that the 
representation of the algebra $\mathfrak{h}$ is required to hold only in the 
sector of the Hilbert space in which the quantum metric is non-degenerate. 
If the span of such states is dense and invariant then checking the algebra on 
this sector is sufficient.  
 
We thus propose that to identify a non degenerate sector in LQG 
becomes an integral part of 
looking for an anomaly free representation of $\mathfrak{h}$. Such a sector can 
hopefully be extracted, e.g., by using renormalisation methods as follows: 
One first 
defines a family of theories at finite resolutions and can identify 
non-degenerate states 
at finite resolutions since the number of degrees of freedom is 
(locally) finite.  Then, if the fixed point family of the renormalisation 
flow exists, 
the family of finite resolution non-degenerate states is consistent and defines 
finite resolution projections of continuum non-degenerate states. These 
qualify as 
the continuum states on which to test the continuum algebra $\mathfrak{h}$.  
  
For the much simpler U(1)$^3$ quantum gravity model we show in our companinon
paper \cite{50} that one can actually complete all the steps of the 
quantisation and establish the anomaly freeness of $\mathfrak{h}$ 
even without renormalisation and without using dual spaces. Quantum 
non-degeneracy is crucial for this to be possible. It is therefore
conceivable that this mechanism also works in realistic quantum gravity.


\begin{thebibliography}{99}

\parskip -5pt

\bibitem{1} R. M. Wald. General Relativity. The University of
Chicago Press, Chicago, 1989

\bibitem{2} Y. Choquet-Bruhat. General Relativity and the Einstein Equations.
Oxford Mathematical Monographs, Oxford University Press, Oxford, 2009.


\bibitem{3} 
C. Palenzuela.
Introduction to Numerical Relativity.
Front. Astron. Space Sci. {\bf 7} (2020) 58. e-Print: 2008.12931 [gr-qc]

\bibitem{4} E. Poisson, C. M. Will. Gravity: Newtonian, Post-Newtonian,
Relativistic. Cambridge University Press, Cambridge 2014.\\
A. Buonanno, B. Sathyaprakash.
Sources of Gravitational Waves: Theory and Observations,
e-Print: 1410.7832.\\
T. Damour. 
Introductory lectures on the Effective One Body formalism.
Int. J. Mod. Phys. {\bf A23} (2008) 1130, e-Print: 0802.4047.\\
R. A. Porto. The effective field theorist’s approach to gravitational 
dynamics.
Phys. Rept. {\bf 633} (2016) 1, e-Print: 1601.04914.
Z. Bern, C. Cheung, R. Roiban, C.-H. Shen, M. P. Solon and M. Zeng.
Black Hole Binary Dynamics from the Double Copy and Effective Theory.
JHEP {\bf 10} (2019) 206, e-Print: 1908.01493


\bibitem{5} C. Rovelli. Quantum Gravity. Cambridge University
Press, Cambridge, 2004.\\
T. Thiemann. Modern Canonical Quantum General Relativity. Cambridge
University Press, Cambridge, 2007\\
J. Pullin, R. Gambini. A first course in Loop Quantum Gravity.
Oxford University Press, New York, 2011\\
C. Rovelli, F. Vidotto. Covariant Loop Quantum Gravity. Cambridge
University Press, Cambridge, 2015

\bibitem{6} A. Ashtekar. New Variables for Classical and Quantum Gravity.
Phys. Rev. Lett. {\bf 57} (1986) 2244-2247\\
J. F. G. Barbero, A real polynomial formulation of general relativity in
terms of connections, Phys. Rev. {\bf D49} (1994) 6935-6938

\bibitem{7} M. Creutz. Quarks, Gluons and Lattices.
Cambridge University Press, Cambridge, 1985.

\bibitem{8} S. A. Hojman, K. Kuchar, C. Teitelboim. Geometrodynamics
Regained. Annals Phys. {\bf 96} (1976) 88-135

\bibitem{9} P.A.M Dirac, Phys. Rev. {\bf 73} (1948) 1092;
Rev. Mod. Phys. {\bf 21} (1949) 392\\
J. A. Wheeler. Geometrodynamics. Academic Press, New York, 1962\\
B. S. DeWitt, Phys. Rev. {\bf 160} (1967) 1113; Phys. Rev.
{\bf 162} (1967) 1195; Phys. Rev. {\bf 162} (1967) 1239.

\bibitem{10} D. Marolf, D. Giulini. On the generality of refined
algebraic quantization.
Class. Quant. Grav. {\bf 16} (1999) 2479-2488; e-Print:
gr-qc/9812024 [gr-qc]

\bibitem{10} T. Thiemann. Quantum spin dynamics. VIII. The Master constraint.
Class. Quant. Grav. {\bf 23} (2006) 2249-2266, [gr-qc/0510011]

\bibitem{11} A. Ashtekar, C.J. Isham. Representations of the Holonomy
Algebras of Gravity and Non-Abelean Gauge Theories.
Class. Quantum Grav. {\bf 9} (1992) 1433, [hep-th/9202053]\\
A. Ashtekar, J. Lewandowski. Representation
theory of analytic Holonomy $C^\star$ algebras. In: Knots and
Quantum Gravity, J. Baez (ed.), Oxford University Press, Oxford 1994\\
A. Ashtekar, J. Lewandowski. Projective Techniques and
Functional Integration for Gauge Theories. J.
Math. Phys. {\bf 36}, 2170 (1995), [gr-qc/9411046]\\
C. Fleischhack. Representations of the Weyl algebra in quantum geometry.
Commun. Math. Phys. {\bf 285} (2009) 67-140, [math-ph/0407006]\\
J. Lewandowski, A. Okolow, H. Sahlmann, T. Thiemann.
Uniqueness of diffeomorphism invariant states on holonomy-flux algebras.
Commun. Math. Phys. {\bf 267} (2006) 703-733, [gr-qc/0504147]


\bibitem{12} T. Thiemann.
Kinematical Hilbert spaces for Fermionic and Higgs quantum field theories.
Class. Quant. Grav. {\bf 15} (1998) 1487-1512; e-Print:
        gr-qc/9705021 [gr-qc]


\bibitem{12a} A. Ashtekar, J. Lewandowski, D. Marolf, J. Mour\~ao, T.
Thiemann. Quantization for diffeomorphism invariant theories
of connections with local degrees of freedom. Journ. Math. Phys.
{\bf 36} (1995) 6456-6493, [gr-qc/9504018]

\bibitem{13} C. Rovelli and L. Smolin.
Discreteness of volume and area in quantum gravity.
Nucl. Phys. {\bf B442} (1995), 593-622; Erratum: Nucl. Phys.
{\bf B456} (1995) 753, [gr-qc/9411005]\\
A. Ashtekar and J. Lewandowski.
Quantum theory of geometry I: Area Operators.
Class. Quant. Grav. {\bf 14} (1997) A55-A82, [gr-qc/9602046];
Quantum theory of geometry II:
Volume operators. Adv. Theo. Math. Phys. {\bf 1} (1997) 388-429,
[gr-qc/9711031]

\bibitem{14} K. Giesel, T. Thiemann. 
Consistency check on volume and triad operator quantisation in loop 
quantum gravity. I.
Class. Quant. Grav. {\bf 23} (2006) 5667-5692, 
e-Print: gr-qc/0507036 [gr-qc];\\
II. Class. Quant. Grav. {\bf 23} (2006) 5693-5772, e-Print: gr-qc/0507037 
[gr-qc]

\bibitem{15} T. Thiemann. Anomaly-free Formulation of non-perturbative,
four-dimensional Lorentzian Quantum Gravity. Physics Letters {\bf B380}
(1996) 257-264, [gr-qc/9606088]\\
T. Thiemann. Quantum Spin Dynamics (QSD).
Class. Quantum Grav. {\bf 15} (1998) 839-73, [gr-qc/9606089];
Quantum Spin Dynamics (QSD) : V.
Quantum Gravity as the Natural Regulator of the Hamiltonian Constraint
of Matter Quantum Field Theories.
Class. Quantum Grav. {\bf 15} (1998) 1281-1314, [gr-qc/9705019]

\bibitem{16} C. Rovelli and L. Smolin.
Spin networks and quantum gravity. Phys. Rev. {\bf D52} (1995) 5743-5759;
e-Print: gr-qc/9505006 [gr-qc]

\bibitem{17} 
T. Thiemann.
QSD 3: Quantum constraint algebra and physical scalar product in 
quantum general relativity. Class. Quant. Grav. {\bf 15} (1998) 1207-1247,
e-Print: gr-qc/9705017 [gr-qc]


\bibitem{18} T. Thiemann. Canonical quantum gravity, constructive QFT and
renormalisation.
Front. in Phys. {\bf 8} (2020) 548232, Front. in Phys. {\bf 0} (2020) 457.
e-Print: 2003.13622 [gr-qc]

\bibitem{19} R. Gambini, J. Lewandowski, D. Marolf, J. Pullin
On the consistency of the constraint algebra in spin network quantum gravity.
Int. J. Mod. Phys. {\bf D 7} (1998) 97-109; e-Print: gr-qc/9710018 [gr-qc]

\bibitem{20} H. Nicolai, K. Peeters, M. Zamaklar.
Loop quantum gravity: An Outside view.
Class. Quant. Grav. {\bf 22} (2005) R193, e-Print: hep-th/0501114 [hep-th]


\bibitem{21} 
A. Laddha. Hamiltonian constraint in Euclidean LQG revisited:
First hints of off-shell Closure. e-Print: 1401.0931 [gr-qc]\\
M. Varadarajan. Euclidean LQG Dynamics: An Electric Shift in Perspective.
Class. Quant. Grav. {\bf 38} (2021) 13, 135020. e-Print: 2101.03115 [gr-qc]\\
M. Varadarajan. Anomaly free quantum dynamics for Euclidean LQG.
e-Print: 2205.10779 [gr-qc]

\bibitem{21a} 
A. Ashtekar, M. Varadarajan.
Gravitational Dynamics—A Novel Shift in the Hamiltonian Paradigm.
Universe 7 (2021) 1, 13. e-Print: 2012.12094 [gr-qc]

\bibitem{50} T. Thiemann. Exact quantisation of U(1)$^3$ 
quantum gravity via exponentiation of the hypersurface deformation 
algebroid.


\bibitem{21b} J. D. Brown, K. V. Kuchar.
Dust as a standard of space and time in canonical quantum gravity.
Phys. Rev. {\bf D51} (1995) 5600-5629.[gr-qc/9409001]\\
K. V. Kuchar, C. G. Torre,
Gaussian reference fluid and interpretation of quantum geometrodynamics.
Phys. Rev. {\bf D43} (1991) 419-441.\\
V. Husain, T. Pawlowski. Time and a physical Hamiltonian for quantum gravity.
Phys.Rev.Lett. 108 (2012) 141301. e-Print: 1108.1145 [gr-qc]\\
M. Domagala, K. Giesel, W. Kaminski, J. Lewandowski.
Gravity quantized: Loop Quantum Gravity with a Scalar Field.
Phys. Rev. {\bf D82} (2010) 104038, [arXiv:1009.2445]\\
K. Giesel, T. Thiemann. Scalar Material Reference Systems and Loop Quantum
Gravity. Class. Quant. Grav. {\bf 32} (2015) 135015,
[arXiv:1206.3807]


\bibitem{22} 
T. Thiemann. Complexifier coherent states for canonical
quantum general relativity.
Class. Quant. Grav. {\bf 23} (2006) 2063-2118,
[gr-qc/0206037]\\
T. Thiemann. Gauge field theory coherent states (GCS): I. General
properties. Class. Quant. Grav. {\bf 18} (2001) 2025-2064,
[hep-th/0005233]\\
T. Thiemann, O. Winkler. Gauge field theory coherent
states
(GCS): II. Peakedness properties. Class. Quant.
Grav. {\bf 18} (2001) 2561-2636, [hep-th/0005237];
Gauge field theory coherent states
(GCS): III. Ehrenfest theorems.
Class. Quant. Grav. {\bf 18} (2001) 4629-4681, [hep-th/0005234]


\bibitem{32} K. Giesel, T. Thiemann. Algebraic quantum gravity
(AQG) I. Conceptual setup. Class. Quant. Grav. {\bf 24} (2007)
2465-2498 [gr-qc/0607099];
Algebraic quantum gravity
(AQG) II. Semiclassical analysis.
Class. Quant. Grav. {\bf 24} (2007) 2499-2564, [gr-qc/0607100];
Algebraic quantum gravity
(AQG) III. Semiclassical perturbation theory.
Class. Quant. Grav. {\bf 24} (2007) 2565-2588, [gr-qc/0607101];
Algebraic quantum gravity (AQG). IV. Reduced phase space quantisation of
loop quantum gravity. Class. Quant. Grav. {\bf 27} (2010) 175009,
[arXiv:0711.0119]


\bibitem{24} Dirac Constraint Quantization of a Parametrized Field Theory by
Anomaly - Free Operator Representations of Space-time Diffeomorphisms.
Phys. Rev. {\bf D 39} (1989) 2263-2280.\\
K. Kuchar. Parametrized Scalar Field on R X S(1):
Dynamical Pictures, Space-time Diffeomorphisms, and Conformal Isometries.
Phys. Rev. {\bf D 39} (1989) 1579-1593

\bibitem{25} M. Varadarajan. Propagation in Polymer Parameterised Field
Theory. Class. Quant. Grav. {\bf 34} (2017) 1, 015012. e-Print:
1609.06034 [gr-qc]\\
A. Laddha, M. Varadarajan.
The Hamiltonian constraint in Polymer Parametrized Field Theory.
Phys.Rev. {\bf D 83} (2011) 025019. e-Print: 1011.2463 [gr-qc]\\
A. Laddha, M. Varadarajan.
Polymer quantization of the free scalar field and its classical limit.
Class. Quant. Grav. {\bf 27} (2010) 175010. e-Print: 1001.3505 [gr-qc]

\bibitem{26} T. Thiemann.
Lessons for Loop Quantum Gravity from Parametrised Field Theory.
e-Print: 1010.2426 [gr-qc]

\bibitem{27} H. Narnhofer, W.E. Thirring.
Covariant QED without indefinite metric.
Rev. Math. Phys. {\bf 4} (1992) spec01, 197-211

\bibitem{28} N. Bodendorfer, T. Thiemann, A. Thurn.
New variables for classical and quantum (super)-gravity in all dimensions.
PoS QGQGS2011 (2011) 022\\
N. Bodendorfer, T. Thiemann, A. Thurn.
New Variables for Classical and Quantum Gravity in all Dimensions 
I. Hamiltonian Analysis. Class. Quant. Grav. {\bf 30} (2013) 045001, e-Print:
1105.3703 [gr-qc];
II. Lagrangian Analysis. Class. Quant. Grav. {\bf 30} (2013) 045002, e-Print:
1105.3704 [gr-qc]; III. Quantum Theory.
Class. Quant. Grav. {\bf 30} (2013) 045003, e-Print: 1105.3705 [gr-qc]




\bibitem{29} T. Thiemann.
Quantum spin dynamics (QSD): 7. Symplectic structures and continuum 
lattice formulations of gauge field theories. Class. Quant.
Grav. {\bf 18} (2001) 3293-3338, e-Print:
        hep-th/0005232 [hep-th]

\bibitem{33} T. Thiemann, E.-A. Zwicknagel. Hamiltonian Renormalisation VI.
Parametrised Field Theory on the cylinder.

\bibitem{30} 
T. Koslowski and H. Sahlmann.
Loop quantum gravity vacuum with nondegenerate geometry.
SIGMA {\bf 8} (2012) 026, e-Print: 1109.4688 [gr-qc]

\bibitem{31}  J. J. Rushanan. On the Vandermonde Matrix. 
The American Mathematical Monthly {\bf 10} (1989), 921-924.

\bibitem{34} A. Vince. Periodicity, Quasiperiodicity, and 
Bieberbach's Theorem on Crystallographic Groups.
The American Mathematical Monthly {\bf 104} (1997), 27-35

\bibitem{35} B. Elizaga Navascues, G. A. Mena Marugan.
Hybrid Loop Quantum Cosmology: An Overview.
Front. Astron. Space {\bf Sci. 8} (2021) 81; e-Print:
2011.04559 [gr-qc]

\bibitem{36} P. Hajicek, K. Kuchar.
Constraint quantization of parametrized relativistic gauge systems in
curved spacetimes. Phys. Rev. D {\bf 41}, 1091

\end{thebibliography}
\end{document}